\documentclass[sigconf]{aamas} 

\usepackage{balance} 

\usepackage{tikz}
\usetikzlibrary{automata,positioning}
\usetikzlibrary{shapes, arrows, automata, positioning}

\usepackage{pgfplots}
\pgfplotsset{compat=1.9}
\usepgfplotslibrary{external}

\usepackage[linesnumbered,ruled,vlined]{algorithm2e}

\usepackage{subcaption}

\setcopyright{ifaamas}
\acmConference[AAMAS '21]{Proc.\@ of the 20th International Conference on Autonomous Agents and Multiagent Systems (AAMAS 2021)}{May 3--7, 2021}{Online}{U.~Endriss, A.~Now\'{e}, F.~Dignum, A.~Lomuscio (eds.)}
\copyrightyear{2021}
\acmYear{2021}
\acmDOI{}
\acmPrice{}
\acmISBN{}

\newtheorem{theorem}{Theorem}

\newtheorem{definition}{Definition}

\acmSubmissionID{581}

\title{Reward Machines for Cooperative \\Multi-Agent Reinforcement Learning}

\author{Cyrus Neary}
\affiliation{
  \institution{The University of Texas at Austin}
  }
\email{cneary@utexas.edu}

\author{Zhe Xu}
\affiliation{
  \institution{Arizona State University}
  }
\email{xzhe1@asu.edu}

\author{Bo Wu}
\affiliation{
  \institution{The University of Texas at Austin}
  }
\email{bowu86@gmail.com}

\author{Ufuk Topcu}
\affiliation{
  \institution{The University of Texas at Austin}
  }
\email{utopcu@utexas.edu}

\begin{abstract}
In cooperative multi-agent reinforcement learning, a collection of agents learns to interact in a shared environment to achieve a common goal. We propose the use of reward machines (RM) --- Mealy machines  used as structured representations of reward functions --- to encode the team's task. The proposed novel interpretation of RMs in the multi-agent setting explicitly encodes required teammate interdependencies, allowing the team-level task to be decomposed into sub-tasks for individual agents. We define such a notion of RM decomposition and present algorithmically verifiable conditions guaranteeing that distributed completion of the sub-tasks leads to team behavior accomplishing the original task. This framework for task decomposition provides a natural approach to decentralized learning: agents may learn to accomplish their sub-tasks while observing only their local state and abstracted representations of their teammates. We accordingly propose a decentralized q-learning algorithm. Furthermore, in the case of undiscounted rewards, we use local value functions to derive lower and upper bounds for the global value function corresponding to the team task. Experimental results in three discrete settings exemplify the effectiveness of the proposed RM decomposition approach, which converges to a successful team policy an order of magnitude faster than a centralized learner and significantly outperforms hierarchical and independent q-learning approaches.
\end{abstract}


\newcommand{\BibTeX}{\rm B\kern-.05em{\sc i\kern-.025em b}\kern-.08em\TeX}

\begin{document}

\newcommand{\Prob}{P}

\newcommand{\metaStates}{S_M}
\newcommand{\commonMetaState}{s_M}

\newcommand{\Goal}{Goal}
\newcommand{\GoalReached}{g}
\newcommand{\RendezvousLocation}{M}
\newcommand{\Rendezvous}{r}
\newcommand{\LeaveRendezvous}{l}

\newcommand{\drone}{D}
\newcommand{\groundTeam}{GT}
\newcommand{\detection}{T}
\newcommand{\droneMove}{M}
\newcommand{\groundTeamMove}{GT}
\newcommand{\noTarget}{N}

\newcommand{\argmin}{\mathrm{arg}\min}
\newcommand{\argmax}{\mathrm{arg}\max}
\newcommand{\argsup}{\mathrm{arg}\sup}
\newcommand{\arginf}{\mathrm{arg}\inf}
\newcommand{\norm}[1]{\left\lVert#1\right\rVert}

\newcommand{\setOfQFunctions}{Q}
\newcommand{\setOfQFunctionsNew}{Q_{\textrm{new}}}            

\newcommand{\set}[1]{\{ #1 \}}
\newcommand{\reals}{\mathbb{R}}

\newcommand{\tabincell}[2]{\begin{tabular}{@{}#1@{}}#2\end{tabular}}

\newcommand{\policy}{\pi}
\newcommand{\JointPolicy}{\boldsymbol{\pi}}
\newcommand{\qValue}{q}
\newcommand{\qValueNew}{q_{\textrm{new}}}
\newcommand{\commonReward}{r}
\newcommand{\maxLengthEpisode}{m}
\newcommand{\learningRate}{\alpha}

\newcommand{\init}{I}

\newcommand{\dfaStates}{V}
\newcommand{\dfaCommonState}{v}
\newcommand{\dfaInitialState}{\dfaCommonState_\init}
\newcommand{\dfaTransition}{\delta}
\newcommand{\dfaInputAlphabet}{E}
\newcommand{\dfaCommonEvent}{e}
\newcommand{\dfaWord}{\omega}
\newcommand{\dfaFinalStates}{\dfaStates_F}
\newcommand{\dfa}{\mathfrak A}
\newcommand{\word}{\omega}
\newcommand{\dfaSymbol}{\ensuremath{b}}

\newcommand{\mdp}{\mathcal{M}}
\newcommand{\mdpStates}{S}
\newcommand{\mdpCommonState}{s}
\newcommand{\mdpCommonStateRV}{S}
\newcommand{\mdpInitialState}{\mdpCommonState_\init}
\newcommand{\mdpActions}{A}
\newcommand{\mdpCommonActionRV}{A}
\newcommand{\mdpCommonAction}{a}
\newcommand{\mdpTransition}{p}
\newcommand{\mdpRewardFunction}{r}
\newcommand{\mdpCommonRewardRV}{R}
\newcommand{\mdpDiscount}{\gamma}
\newcommand{\mdpTrajectory}{\zeta}
\newcommand{\mdpLabel}{\ell}
\newcommand{\mdpCommonReward}{r}
\newcommand{\trajectory}[1]{\ensuremath{\mdpCommonState_0 \mdpCommonAction_0\ldots \mdpCommonState_#1 \mdpCommonAction_#1 \mdpCommonState_{#1 + 1}}}
\newcommand{\mdpValue}{V}

\newcommand{\SG}{\mathcal{G}}
\newcommand{\SGJointStates}{\mathcal{S}}
\newcommand{\SGCommonJointState}{\mathbf{s}}
\newcommand{\SGStates}{\mdpStates}
\newcommand{\SGCommonState}{s}
\newcommand{\SGRewardFunction}{R}
\newcommand{\SGDiscount}{\gamma}
\newcommand{\SGActions}{A}
\newcommand{\SGCommonAction}{a}
\newcommand{\SGJointActions}{\mathcal{A}}
\newcommand{\SGCommonJointAction}{\mathbf{a}}
\newcommand{\SGTransition}{p}
\newcommand{\SGLocalTransition}{p}
\newcommand{\SGAgentSet}{\mathcal{N}}
\newcommand{\SGNumAgents}{N}

\newcommand{\distribution}{\Delta}

\newcommand{\possibleTeamStates}{B}
\newcommand{\eventAgentIndex}{I_{\RMCommonEvent}}

\newcommand{\RM}{\mathcal{R}}
\newcommand{\RMStates}{U}
\newcommand{\RMCommonState}{u}
\newcommand{\RMInitialState}{\RMCommonState_I}
\newcommand{\RMLabels}{\RMEvents^*}
\newcommand{\RMPropositions}{\mathcal{P}}
\newcommand{\RMLabelingFunction}{L}
\newcommand{\RMCommonLabel}{l}
\newcommand{\RMTransition}{\delta}
\newcommand{\RMOutput}{\sigma}
\newcommand{\RMCommonReward}{r}
\newcommand{\RMFinalStates}{F}
\newcommand{\RMEvents}{\Sigma}
\newcommand{\RMCommonEvent}{e}
\newcommand{\RMString}{l}
\newcommand{\RMCommonInput}{l}
\newcommand{\EmptyString}{\varepsilon}
\newcommand{\Projection}{P}
\newcommand{\RMEventSequence}{\xi}

\newcommand{\qValueAgenti}{\qValue_{\RMCommonState^i}}

\newcommand{\SGStateSubset}{B}

\newcommand{\LabelingFunctionOutput}{2^\RMEvents}

\newcommand{\localLabelingFunctionOutput}{2^{\RMEvents_i}}
\newcommand{\commonLocalLabelingFunctionOutput}{l}
\newcommand{\syncOutput}{\Tilde{\commonLocalLabelingFunctionOutput}}

\newcommand{\machine}{\mathcal A}
\newcommand{\mealyStates}{V}

\newcommand{\mealyCommonState}{v}
\newcommand{\mealyCommonInput}{u}        
\newcommand{\mealyInputAlphabet}{\Sigma}
\newcommand{\mealyInit}{{\mealyCommonState_\init}}
\newcommand{\mealyInitHat}{{\hat{\mealyCommonState}_\init}}
\newcommand{\mealyOutputAlphabet}{\ensuremath{\mathbb{R}}}
\newcommand{\mealyOutput}{\sigma}
\newcommand{\mealyTransition}{\delta}
\newcommand{\inputTrace}{\lambda}
\newcommand{\outputTrace}{\rho}

\newcommand{\mealyOutputNew}{\mealyOutput_{\textrm{new}}}
\newcommand{\mealyTransitionNew}{\mealyTransition_{\textrm{new}}}
\newcommand{\mealyStatesNew}{\mealyStates_{\textrm{new}}}
\newcommand{\mealyCommonStateNew}{\mealyCommonState_{\textrm{new}}}
\newcommand{\mealyInitNew}{\mealyInit_{\textrm{new}}}

\newcommand{\Relation}{\sim}
\newcommand{\projectionRelation}{\Relation_{\RMEvents_i}}

\newcommand{\episode}{n}
\newcommand{\LINEFORALL}[3][default]{%
  \ALC\algorithmicforall\ #2\ \algorithmicdo%
  \ALC\ #3\ \algorithmicendfor%
}

\newcommand{\commFunction}{C}
\newcommand{\buttonSize}{3.5pt}

\newcommand{\buttonsAgentOne}{A_1}
\newcommand{\buttonsAgentTwo}{A_2}
\newcommand{\buttonsAgentThree}{A_3}

\newcommand{\redButton}{R_B}
\newcommand{\yellowButton}{Y_B}
\newcommand{\greenButton}{G_B}
\newcommand{\redTile}{R_T}
\newcommand{\yellowTile}{Y_T}
\newcommand{\greenTile}{G_T}
\newcommand{\agentTwoRedButton}{\buttonsAgentTwo^{\redButton}}
\newcommand{\agentTwoNotRedButton}{\buttonsAgentTwo^{\lnot \redButton}}
\newcommand{\agentThreeRedButton}{\buttonsAgentThree^{\redButton}}
\newcommand{\agentThreeNotRedButton}{\buttonsAgentThree^{\lnot \redButton}}
\newcommand{\goal}{Goal}

\pagestyle{fancy}
\fancyhead{}


\maketitle 


\section{Introduction}

In multi-agent reinforcement learning (MARL), a collection of agents learn to maximize expected long-term return through interactions with each other and with a shared environment. We study MARL in a cooperative setting: all of the agents are rewarded collectively for achieving a team task. 

Two challenges inherent to MARL are coordination and non-stationarity. Firstly, coordination is needed between agents because the correctness of any individual agent's actions may depend on the actions of its teammates \cite{boutilier1996planning, guestrin2002coordinated}. Secondly, the agents are learning and updating their behaviors simultaneously. From the point of view of any individual agent, the learning problem is non-stationary; the best solution for any individual is constantly changing \cite{non_stationary_MARL_survey}. 

A reward machine (RM) is a Mealy machine used to define tasks and behaviors dependent on abstracted descriptions of the environment \cite{icarte2018using}. Intuitively, RMs allow agents to separate tasks into stages and to learn different sets of behaviors for the different portions of the overall task. In this work, we use RMs to describe cooperative tasks and we introduce a notion of RM decomposition for the MARL problem. The proposed use of RMs explicitly encodes the information available to each agent, as well as the teammate communications necessary for successful cooperative behavior. The global (cooperative) task can then be decomposed into a collection of new RMs, each encoding a sub-task for an individual agent. We propose a decentralized learning algorithm that trains the agents individually using these sub-task RMs, effectively reducing the team's task to a collection of single-agent reinforcement learning problems. The algorithm assumes each agent may only observe its own local state and the information encoded in its sub-task RM.

Furthermore, we provide conditions guaranteeing that if each agent accomplishes its sub-task, the corresponding joint behavior accomplishes the team task. Finally, decomposition of the team's task allows for each agent to be trained independently of its teammates and thus addresses the problems posed by non-stationarity. Individual agents  condition their actions on abstractions of their teammates, eliminating the need for simultaneous learning. 

Experimental results in three discrete domains exemplify the strengths of the proposed decentralized algorithm. In a two-agent rendezvous task, the proposed algorithm converges to successful team behavior more than an order of magnitude faster than a centralized learner and performs roughly on par with hierarchical independent learners (h-IL) \cite{hierarchicalDeepMARL}. In a ten-agent variant of the task, the proposed algorithm quickly learns effective team behavior while neither h-IL nor independent q-learning (IQL) \cite{MARLIndependentvsCooperativeAgentsTan} converges to policies completing the task within the allowed training steps. 
\section{Preliminaries}

A Markov decision process (MDP) is a tuple $\mdp = \langle \mdpStates, \mdpActions, \mdpRewardFunction, \mdpTransition, \mdpDiscount \rangle$ consisting of a finite set of states $\mdpStates$, a finite set of actions $\mdpActions$, a reward function $\mdpRewardFunction : \mdpStates \times \mdpActions \times \mdpStates \rightarrow \mathbb{R}$, a transition probability function $\mdpTransition : \mdpStates \times \mdpActions \rightarrow \distribution(\mdpStates)$, and a discount factor $\mdpDiscount \in (0,1]$. Here $\distribution(\mdpStates)$ is the set of all probability distributions over 
$\mdpStates$. We denote by $\mdpTransition(\mdpCommonState' | \mdpCommonState, \mdpCommonAction)$ the probability of transitioning to state $\mdpCommonState'$ from state $\mdpCommonState$ under action $\mdpCommonAction$. A stationary policy $\policy : \mdpStates \rightarrow \distribution(\mdpActions)$ maps states to probability distributions over the set of actions.  In particular, if an agent is in state $\mdpCommonState_t \in \mdpStates$ at time step $t$ and is following policy $\policy$, then $\policy(\mdpCommonAction_t | \mdpCommonState_t)$ denotes its probability of taking action $\mdpCommonAction_t \in \mdpActions$. 

The goal of reinforcement learning (RL) is to learn an optimal policy $\policy^*$ maximizing the expected sum of discounted future rewards from any state \cite{sutton2018reinforcement}. The q-function for policy $\policy$ is defined as the expected discounted future reward that results from taking action $\mdpCommonAction$ from state $\mdpCommonState$ and following policy $\policy$ thereafter. Tabular q-learning \cite{watkins1992q}, an RL algorithm, uses the experience $\{(\mdpCommonState_t, \mdpCommonAction_t, \mdpCommonReward_t, \mdpCommonState_{t+1})\}_{t\in \mathbb{N}_0}$ of an agent interacting with an MDP to learn the q-function $\qValue^*(\mdpCommonState, \mdpCommonAction)$ corresponding to an optimal policy $\policy^*$. Given $\qValue^*(\mdpCommonState, \mdpCommonAction)$, the optimal policy may be recovered.

A common framework used to extend RL to a multi-agent setting is the Markov game \cite{littman1994markov,boutilier1996planning}. 
A cooperative Markov game of $\SGNumAgents$ agents is a tuple $\SG = \langle \SGStates_1, ..., \SGStates_\SGNumAgents, \SGActions_1, ..., \SGActions_{\SGNumAgents}, \SGTransition, \SGRewardFunction, \SGDiscount \rangle$. $\SGStates_i$ and $\SGActions_i$ are the finite sets of agent $i$'s local states and actions respectively. We define the set of joint states as $\SGJointStates = \SGStates_1 \times ... \times \SGStates_N$ and we similarly define the set of joint actions to be $\SGJointActions = \SGActions_1 \times ... \times \SGActions_N$. $\SGTransition : \SGJointStates \times \SGJointActions \rightarrow \distribution(\SGJointStates)$ is a joint state transition probability distribution. $\SGRewardFunction : \SGJointStates \times \SGJointActions \times \SGJointStates \rightarrow \mathbb{R}$ is the team's collective reward function which is shared by all agents, and $\SGDiscount \in (0,1]$ is a discount factor.

In this work, we assume the dynamics of each agent are independently governed by local transition probability functions $\SGLocalTransition_i : \SGStates_i \times \SGActions_i \rightarrow \distribution(\SGStates_i)$. The joint transition function is then constructed as $\SGTransition(\SGCommonJointState' | \SGCommonJointState, \SGCommonJointAction) = \Pi_{i=1}^\SGNumAgents \SGLocalTransition_i(\SGCommonState_i' | \SGCommonState_i, \SGCommonAction_i)$, for all $\SGCommonJointState$, $\SGCommonJointState' \in \SGJointStates$ and $\SGCommonJointAction \in \SGJointActions$. A team policy is defined as $\JointPolicy : \SGJointStates \rightarrow \distribution(\SGJointActions)$. 
As in the single agent case, the objective of team MARL is to find a team policy $\JointPolicy^*$ maximizing expected discounted future reward from any joint state.
\section{Reward Machines for MARL}
\label{sec:reward_machines_for_MARL}

To introduce reward machines (RM) and to illustrate how they may be used to encode a team's task, we consider the example shown in Figure \ref{fig:buttons_gridworld}. Three agents, denoted $\buttonsAgentOne$, $\buttonsAgentTwo$, and $\buttonsAgentThree$, operate in a shared environment with the objective of allowing $\buttonsAgentOne$ to reach the goal location denoted $\Goal$. However, the red, yellow, and green colored regions are blocking the paths of agents $\buttonsAgentOne$, $\buttonsAgentTwo$, and $\buttonsAgentThree$ respectively. To allow the agents to cross these colored regions, the button of the corresponding color must be pressed first. Furthermore, the yellow and green buttons may be pressed by an individual agent, but the red button requires two agents to simultaneously occupy the button's location before it is activated. The dashed and numbered arrows in the figure illustrate the sequence of events necessary for task completion: $\buttonsAgentOne$ should push the yellow button allowing $\buttonsAgentTwo$ to proceed to the green button, which is necessary for $\buttonsAgentThree$ to join $\buttonsAgentTwo$ in pressing the red button, finally allowing $\buttonsAgentOne$ to cross the red region and reach $\Goal$.

\begin{figure*}[t!]
    \centering
    \begin{subfigure}[t]{0.3\textwidth}
        \centering
        \resizebox{0.8\textwidth}{!}{

\begin{tikzpicture}[scale=1.0]

\filldraw[fill=black!0!white] (0,0) rectangle (10,10);
\filldraw[fill=black, draw=black] (3,2) rectangle (4,10);
\filldraw[fill=black, draw=black] (4,2) rectangle (10,3);
\filldraw[fill=black, draw=black] (7,5) rectangle (8,10);
\filldraw[fill=yellow, opacity=0.6, draw=black] (4,6) rectangle (7,8);
\filldraw[fill=green, opacity=0.6, draw=black] (8,6) rectangle (10,8);
\filldraw[fill=red, opacity=0.6, draw=black] (5,0) rectangle (8,2);


\filldraw[fill=yellow, opacity=0.6, draw=black] (2.3, 9.3) circle (0.5cm);
\filldraw[fill=green, opacity=0.6, draw=black] (6.5, 4.5) circle (0.5cm);
\filldraw[fill=red, opacity=0.6, draw=black] (9.3, 3.7) circle (0.5cm);
\node at (2.3, 9.3) {\Huge $\yellowButton$};
\node at (6.5, 4.5) {\Huge $\greenButton$};
\node at (9.3, 3.7) {\Huge $\redButton$};

\node[anchor=center] at (9.0, 1.5) (g) {\Huge $\goal$};
\node[anchor=center] at (0.5, 9.3) (a1) {\Huge $A_1$};
\node[anchor=center] at (5.5, 9.5) (a2) {\Huge $A_2$};
\node[anchor=center] at (8.5, 9.5) (a3) {\Huge $A_3$};

\draw[draw=black, very thick] (0,0) rectangle (10,10);
\path[draw, dashed, ->, ultra thick] (a1.0) -- node[below]{\Huge 1} (1.8, 9.3);
\path[draw, dashed, ultra thick] (a2.270) -- node [left, near end]{\Huge 2} (5.5, 4.5);
\path[draw, dashed, ->, ultra thick] (5.5,4.5) -- (6, 4.5);
\path[draw, dashed, ultra thick] (6.5, 4) -- (6.5, 3.7);
\path[draw, dashed,->, ultra thick] (6.5, 3.7) -- node[above]{\Huge 3} (8.8, 3.7);
\path[draw, dashed, ultra thick] (9, 9.5) -- (9.3, 9.5);
\path[draw, dashed, -> , ultra thick] (9.3, 9.5) -- node[left, near end]{\Huge 3} (9.3, 4.2);
\path[draw, dashed, ultra thick] (2.3, 8.8) -- node[left]{\Huge 4} (2.3, 1.5);
\path[draw, dashed, ->, ultra thick] (2.3, 1.5) -- (8.3, 1.5);

\end{tikzpicture}
}
        \caption{Cooperative buttons domain.}
        \label{fig:buttons_gridworld}
    \end{subfigure}%
    ~
    \begin{subfigure}[t]{0.7\textwidth}
        \centering
        \tikzset{auto,
        ->,
        >=stealth,
        node distance=2.2cm,
        every node/.style={scale=0.85, minimum size=0pt, inner sep=0pt}}

\newcommand{\blockOpacity}{15}
\newcommand{\minBlockHeight}{0.5cm}
\newcommand{\tileDim}{0.5cm}

\usetikzlibrary{shapes.geometric}

\newcommand{\hDist}{1.5cm}
\newcommand{\vDist}{1.0cm}

\resizebox{0.9\textwidth}{!}{

\begin{tikzpicture}
    \node[state, initial above] (ui) {$\RMInitialState$};
    \path (ui.0)+(\hDist, 0.0cm) node (u1) [state] {$\RMCommonState_1$};
    \path (u1.0)+(\hDist, 0.0cm) node (u2) [state] {$\RMCommonState_2$};
    \path (u2.90)+(\hDist, \vDist) node (u3) [state] {$\RMCommonState_3$};
    \path (u2.270)+(\hDist, -\vDist) node (u4) [state] {$\RMCommonState_4$};
    \path (u3.270)+(\hDist, -\vDist) node (u5) [state] {$\RMCommonState_5$};
    \path (u5.0)+(\hDist, 0.0cm) node (u6) [state] {$\RMCommonState_6$};
    \path (u6.270)+(0.0cm, -\vDist) node (u7) [state, accepting] {$\RMCommonState_7$};
    
    \draw (ui) -> node[above, yshift=0.1cm] {$\yellowButton$} (u1);
    \draw (u1) edge node[above, yshift=0.1cm] {$\greenButton$} (u2);
    \draw (u2) edge node[below, xshift=0.2cm] {$\agentTwoRedButton$} (u3);
    \draw (u2) edge node[above, xshift=0.2cm] {$\agentThreeRedButton$} (u4);
    \draw (u3) edge node[below, xshift=-0.2cm] {$\agentThreeRedButton$} (u5);
    \draw (u4) edge node[above, xshift=-0.2cm] {$\agentTwoRedButton$} (u5);
    \draw (u3) edge[bend right] node[left, xshift=-0.1cm] {$\agentTwoNotRedButton$} (u2);
    \draw (u4) edge[bend left] node[left, xshift=-0.1cm] {$\agentThreeNotRedButton$} (u2);
    \draw (u5) edge node[above, yshift=0.1cm] {$\redButton$} (u6);
    \draw (u5) edge[bend right] node[right, xshift=0.15cm] {$\agentThreeNotRedButton$} (u3);
    \draw (u5) edge[bend left] node[right, xshift=0.1cm] {$\agentTwoNotRedButton$} (u4);
    \draw (u6) edge node[left, xshift=-0.1cm] {$\goal$} (u7);
    
    
    
    
\end{tikzpicture}}
        \caption{Reward machine encoding the cooperative buttons task.}
        \label{fig:buttons_rm}
    \end{subfigure}
    \caption{The multi-agent buttons task. In Figure (a), the colored circles denote the locations of the buttons, the thick black areas are walls the agents cannot cross, and the numbered dotted lines show the order of high-level steps necessary to complete the task. The set of events of the RM in (b) is $\RMEvents = \{ \yellowButton, \greenButton, \redButton, \agentTwoRedButton, \agentTwoNotRedButton, \agentThreeRedButton, \agentThreeNotRedButton, \goal \}$.}
    \label{fig:buttons_task_description}
\end{figure*}
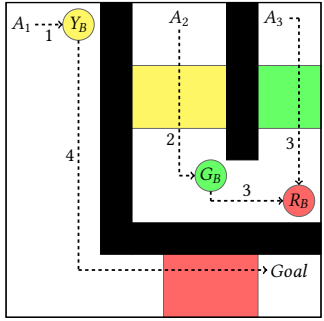
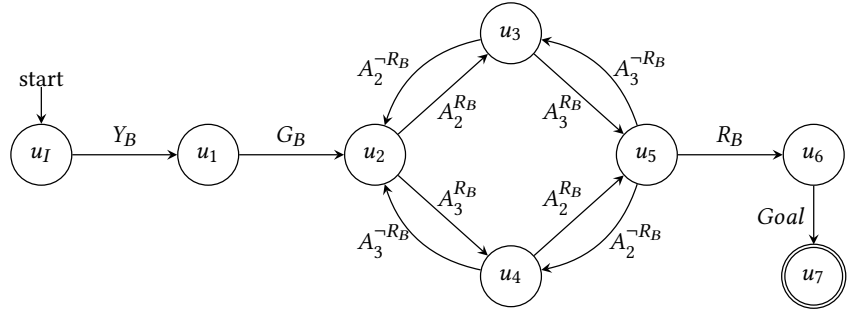


\subsection{Reward Machines for Task Description}
\label{sec:reward_machines}

\begin{definition}
\label{def:rewardMachine}
A \emph{reward machine (RM)} 
$\RM = \langle \RMStates, \RMInitialState, \RMEvents, \RMTransition, \RMOutput \rangle$ consists of 
a finite, nonempty set $\RMStates$ of states, 
an initial state $\RMInitialState \in \RMStates$, 
a finite set $\RMEvents$ of environment events,
a transition function $\RMTransition \colon \RMStates \times \RMEvents \to \RMStates$, 
and an output function $\RMOutput \colon \RMStates \times \RMStates \to \mathbb{R}$.
\end{definition}

Reward machines are a type of Mealy machine used to define temporally extended tasks and behaviors. Initially introduced in \cite{icarte2018using}, we adapt the RM's definition slightly to better suit the multi-agent setting. In particular, we define the transition function to take events $\RMCommonEvent \in \RMEvents$ as input, instead of collections of atomic propositions. This change simplifies the notation required for task decomposition and it is relatively minor; any collection of atomic propositions can be redefined as an individual event $\RMCommonEvent$. Furthermore, the output function $\RMOutput$ is defined to map transitions to constant values, instead of to reward functions. Finally, we note that $\RMTransition$ and $\RMOutput$ are partial functions; they are defined on subsets of $\RMStates \times \RMEvents$ and $\RMStates \times \RMStates$ respectively. Figure \ref{fig:buttons_rm} illustrates the RM encoding the buttons task. 

Set $\RMEvents$ is the collection of all the high-level events necessary to describe the team's task. For example, the event set corresponding to the buttons task from Figure \ref{fig:buttons_gridworld} is \(\RMEvents = \{ \yellowButton\), \(\greenButton\), \(\agentTwoRedButton\), \(\agentTwoNotRedButton\), \(\agentThreeRedButton\), \(\agentThreeNotRedButton\), \(\redButton\), \(\goal \}\).  Here, events \(\yellowButton\), \(\greenButton\), and \(\redButton\) correspond to the yellow, green, or red button being pressed, respectively. Because both agents $\buttonsAgentTwo$ and $\buttonsAgentThree$ must simultaneously stand on the red button for the event $\redButton$ to occur, we additionally include the events $\agentTwoRedButton, \agentTwoNotRedButton, \agentThreeRedButton, \agentThreeNotRedButton$ in $\RMEvents$, which represent $\buttonsAgentTwo$ or $\buttonsAgentThree$ individually either pressing or not pressing the red button. The event \(\goal\) corresponds to agent \(\buttonsAgentOne\) successfully reaching the goal location.

The states $\RMCommonState \in \RMStates$ of the RM represent different stages of the team's task. Transitions between RM states are triggered by events from $\RMEvents$. For example, the buttons task starts in state $\RMInitialState$. From this state, no buttons have been pressed and so the colored regions cannot be entered. When \(\buttonsAgentOne\) presses the yellow button in the environment, the event \(\yellowButton\) will cause the RM to transition to state \(\RMCommonState_1\). From this state, only event \(\greenButton\) causes an outgoing transition; \(\buttonsAgentTwo\) may now proceed across the yellow region in order to press the green button. Because, for example, the red region still prevents \(\buttonsAgentOne\) from reaching the goal location, event \(\goal\) does not cause a transition from \(\RMCommonState_1\). In this way, transitions in the RM represent progress through the task. 

Output function $\RMOutput$ assigns reward values to these transitions. In this work, we restrict ourselves to \textit{task completion} RMs, similar to those studied in \cite{xu2019joint}; $\RMOutput$ should only reward transitions that result in the immediate completion of the task. To formalize this idea, we define a subset  $\RMFinalStates \subseteq \RMStates$ of reward states. If the RM is in a state belonging to $\RMFinalStates$, it means the task is complete. The output function is then defined such that $\RMOutput(\RMCommonState, \RMCommonState') = 1$ if $\RMCommonState \notin \RMFinalStates$ and $\RMCommonState' \in \RMFinalStates$, and is defined to output $0$ otherwise. Furthermore, there should be no outgoing transitions from states in $\RMFinalStates$. In Figure $\ref{fig:buttons_rm}$, $\RMFinalStates = \{\RMCommonState_7\}$. If the event $\goal \in \RMEvents$ occurs while the RM is in state $\RMCommonState_6$ the task has been successfully completed; the RM will transition and return reward 1. 

A run of RM $\RM$ on the sequence of events $\RMCommonEvent_0 \RMCommonEvent_1 ... \RMCommonEvent_k \in \RMEvents^*$ is a sequence $\RMCommonState_0\RMCommonEvent_0\RMCommonState_1\RMCommonEvent_1...\RMCommonState_k\RMCommonEvent_k\RMCommonState_{k+1}$, where $\RMCommonState_0 = \RMInitialState$ and $\RMCommonState_{t+1} = \RMTransition(\RMCommonState_t, \RMCommonEvent_t)$. If $u_{k+1} \in \RMFinalStates$, then $\RMOutput(\RMCommonState_k, \RMCommonState_{k+1}) = 1$. In this case we say that the event sequence $\RMCommonEvent_0...\RMCommonEvent_k$ completes the task described by $\RM$, and we denote this statement $\RM(\RMCommonEvent_0...\RMCommonEvent_k) = 1$. Otherwise, $\RM(\RMCommonEvent_0...\RMCommonEvent_k) = 0$. For example, $\RM(\yellowButton \greenButton \agentTwoRedButton \agentThreeRedButton \redButton \Goal) = 1$, but $\RM(\yellowButton \greenButton \agentTwoRedButton) = 0$.

For notational convenience, we define the transition from state $\RMCommonState$ under an event sequence $\RMCommonEvent_0 \RMCommonEvent_1 ... \RMCommonEvent_k \in \RMEvents^*$ using the recursive definitions $\RMTransition(\RMCommonState, \EmptyString) = \RMCommonState$ and $\RMTransition(\RMCommonState, \RMEventSequence \RMCommonEvent) = \RMTransition(\RMTransition(\RMCommonState, \RMEventSequence), \RMCommonEvent)$ for $\RMEventSequence \in \RMEvents^*$ and $\RMCommonEvent \in \RMEvents$. Here, $\RMEvents^*$ is the Kleene closure of $\RMEvents$ and $\EmptyString$ is the empty string of events. So, $\RM(\RMCommonEvent_0 \RMCommonEvent_1 ... \RMCommonEvent_k) = 1$ if and only if $\RMTransition(\RMInitialState, \RMCommonEvent_0 \RMCommonEvent_1 ... \RMCommonEvent_k) \in \RMFinalStates$.

\subsection{Labeling Functions and Q-Learning with Reward Machines}
\label{sec:labeling_functions}

RMs may be applied to RL problems by using them to replace the reward function in an MDP. However, RMs describe tasks in terms of abstract events. To allow an RM to interface with the underlying environment, we define a \textit{labeling function} $\RMLabelingFunction : \mdpStates \times \RMStates \rightarrow \LabelingFunctionOutput$, which abstracts the current environment state to sets of high-level events. Note that $\RMLabelingFunction$ takes the current RM state $\RMCommonState \in \RMStates$ as well as the environment state $\SGCommonState \in \SGStates$ as input, allowing the events output by $\RMLabelingFunction$ to depend not only on the environment, but also on the current progress through the task. This component of $\RMLabelingFunction$'s definition aides in specifying local labeling functions, discussed in \S \ref{sec:local_labeling_functions}. Also, $\RMLabelingFunction$ is defined to output collections of events, allowing it to capture scenarios in which multiple events occur concurrently. In such a scenario, the events are passed as a sequence to the RM in no particular order.

Q-learning with RMs (QRM) \cite{icarte2018using} is an algorithm that learns a collection of q-functions, one for each RM state $\RMCommonState \in \RMStates$, corresponding to the optimal policies for each stage of the task. Algorithm \ref{alg:QRM} details the method. At each time step, if the agent is in RM state $\RMCommonState_1$ and environment state $\SGCommonState_1$, it uses its estimate of $\qValue_{\RMCommonState_1}(\SGCommonState_1, \cdot)$ to select action $\SGCommonAction$. The environment accordingly progresses to state $\SGCommonState_{2}$. The events output by $\RMLabelingFunction(\SGCommonState_{2}, \RMCommonState_1)$ then cause the RM to transition to state $\RMCommonState_{2}$, and the corresponding reward output by $\RMOutput$ is used to update the estimate of  $\qValue_{\RMCommonState_1}(\SGCommonJointState_1, \SGCommonAction)$: the optimal q-function for RM state $\RMCommonState_1$. At this stage, the agent also queries the rewards and RM transitions that would have occurred had the RM been in any other state $\RMCommonState$. This counterfactual information is used to update the estimate of each q-function $q_{\RMCommonState}$. The tabular QRM algorithm is guaranteed to converge to an optimal policy \cite{icarte2018using}.

A naive approach to applying RMs in the MARL setting would be to treat the entire team as a single agent and to use QRM to learn a centralized policy. This approach quickly becomes intractable, however, due to the exponential scaling of the number of states and actions with the number of agents. Furthermore, it assumes agents communicate with a central controller at every time step, which may be undesirable from an implementation standpoint.

\begin{algorithm}
    \DontPrintSemicolon 
    \KwIn{$\RM = \langle \RMStates, \RMInitialState, \RMEvents, \RMTransition, \RMOutput, \RMFinalStates \rangle$, $\RMLabelingFunction$, $\gamma$, $\alpha$}
    \KwOut{$Q = \{\qValue_{\RMCommonState} : \mdpStates \times \mdpActions \to \mathbb{R} | \RMCommonState \in \RMStates\}$}
    $Q \gets InitializeQFunctions()$\;
    \For{$n = 1$ \textbf{to} $NumEpisodes$} {
        $\RMCommonState_1 \gets \RMInitialState$, $\mdpCommonState_1 \gets environmentInitialState()$\;
        \For{$t=0$ \textbf{to} $NumSteps$}{
        $\mdpCommonAction \gets getAction(\qValue_{\RMCommonState_1}, \mdpCommonState_1)$\;
        $\mdpCommonState_2 \gets executeAction(\mdpCommonState_1, \mdpCommonAction)$\;
        
        $r, \RMCommonState_2 \gets rewardMachineOutput(\RMCommonState_1, \RMLabelingFunction(\SGCommonState_2, \RMCommonState_1))$\;
        
        $\qValue_{\RMCommonState_1}(\mdpCommonState_1, \mdpCommonAction) \gets (1-\alpha)\qValue_{\RMCommonState_1}(\mdpCommonState_1, \mdpCommonAction) + \alpha(r + \gamma \max_{\mdpCommonAction' \in \mdpActions}\qValue_{\RMCommonState_2}(\mdpCommonState_2, \mdpCommonAction'))$\;

        \For{$\RMCommonState \in \RMStates$, $\RMCommonState \neq \RMCommonState_1$}{
            $r, \RMCommonState' \gets rewardMachineOutput(\RMCommonState, \RMLabelingFunction(\SGCommonState_2, \RMCommonState))$;
            
            $\qValue_{\RMCommonState}(\mdpCommonState_1, \mdpCommonAction) \gets (1-\alpha)\qValue_{\RMCommonState}(\mdpCommonState_1, \mdpCommonAction) + \alpha(r + \gamma \max_{\mdpCommonAction' \in \mdpActions}\qValue_{\RMCommonState'}(\mdpCommonState_2, \mdpCommonAction'))$\;
        }
        
        $\RMCommonState_1 \gets \RMCommonState_2$, $\mdpCommonState_1 \gets \mdpCommonState_2$\;
        
        \If{$\RMCommonState_1 \in \RMFinalStates$}{$break$\;}
        }
    }
    \Return{$Q$}\;
    \caption{Q-Learning with Reward Machines}
    \label{alg:QRM}
\end{algorithm}


\subsection{Team Task Decomposition}

\label{sec:team_task_decomp}

A decentralized approach to MARL treats the agents as individual decision-makers, and therefore requires further consideration of the information available to each agent. In this work, we assume the $i^{th}$ agent can observe its own local state $\SGCommonState_i \in \SGStates_i$, but not the local states of its teammates. Given an RM $\RM$ describing the team's task and the corresponding event set $\RMEvents$, we assign the $i^{th}$ agent a subset $\RMEvents_i \subseteq \RMEvents$ of events. These events represent the high-level information that is available to the agent. We call $\RMEvents_i$ the \textit{local event set} of agent $i$. We assume that all events represented in $\RMEvents$ belong to the local event set of at least one of the agents, thus $\bigcup_{i=1}^{\SGNumAgents}\RMEvents_i = \RMEvents$.

For example, in the three-agent buttons task, the local event set assigned to $\buttonsAgentOne$ is $\RMEvents_1 = \{\yellowButton, \redButton, \Goal\}$; $\buttonsAgentOne$ has access to the yellow button, must know when the red button has been pressed, and should eventually proceed to the goal location. Note, for example, that events $\agentTwoRedButton$ and $\agentTwoNotRedButton$ are not included in $\RMEvents_1$ because they are specific to agent \(\buttonsAgentTwo\); they are not necessary pieces of information for the completion of \(\buttonsAgentOne\)'s sub-task. Similarly, event \(\greenButton\) is not in \(\RMEvents_1\) because from the perspective of \(\buttonsAgentOne\)'s task, it is also only an intermediate step in the process of pressing the red button. The event sets of $\buttonsAgentTwo$ and $\buttonsAgentThree$ are $\RMEvents_2 = \{\yellowButton, \greenButton, \agentTwoRedButton, \agentTwoNotRedButton, \redButton\}$ and $\RMEvents_3 = \{\greenButton, \agentThreeRedButton, \agentThreeNotRedButton, \redButton\}$, respectively. 

Extending the definition of natural projections on automata \cite{karimadini2011guaranteed, wong1998complexity} to reward machines, for each agent $i$, we define a new RM, $\RM_i$ $=$ $\langle \RMStates_i, \RMInitialState^i, \RMEvents_i, \RMTransition_i, \RMOutput_i, \RMFinalStates_i \rangle$, called the projection of $\RM$ onto $\RMEvents_i$. We begin by defining a notion of state equivalence under event set $\RMEvents_i$ 

Given $\RM = (\RMStates, \RMInitialState, \RMEvents, \RMTransition, \RMOutput, \RMFinalStates)$ and $\RMEvents_i \subseteq \RMEvents$ we define the equivalence relation on states under event set $\RMEvents_i$, denoted $\Relation_{i} \subseteq \RMStates \times \RMStates$, as the smallest equivalence relation such that:
\begin{enumerate}
    \item For all $\RMCommonState_1, \RMCommonState_2 \in \RMStates$, and $\RMCommonEvent \in \RMEvents$,\\ if $\RMCommonState_2 = \RMTransition(\RMCommonState_1, \RMCommonEvent)$ and $\RMCommonEvent \notin \RMEvents_i$, then $(\RMCommonState_1, \RMCommonState_2) \in \Relation_{i}$.
    \item If $\RMCommonState_1, \RMCommonState_1' \in \RMStates$, $(\RMCommonState_1, \RMCommonState_1') \in \Relation_i$ and $\RMTransition(\RMCommonState_1, \RMCommonEvent) = \RMCommonState_2$, $\RMTransition(\RMCommonState_1', \RMCommonEvent) = \RMCommonState_2'$ are both defined for some $\RMCommonEvent\in \RMEvents_i$, then $(\RMCommonState_2, \RMCommonState_2') \in \Relation_i$.
\end{enumerate}

The equivalence class of any state $\RMCommonState \in \RMStates$ under equivalence relation $\Relation_{i}$ is $[\RMCommonState]_{i} = \{\RMCommonState'\in \RMStates | (\RMCommonState, \RMCommonState') \in \Relation_{i}\}$. The quotient set of $\RMStates$ by $\Relation_{i}$ is defined as the set of all equivalence classes $\RMStates/{\Relation_{i}} = \{[\RMCommonState]_{i} | \RMCommonState \in \RMStates\}$. Equivalence relation $\Relation_i$ may be computed by first finding the smallest equivalence relation satisfying condition $(1)$, and then applying at most $|\RMStates|$ state-merging steps in order to satisfy $(2)$. An algorithm to compute this equivalence relation with runtime \(\mathcal{O}(|\RMStates|^7|\RMEvents_i|^2 + |\RMStates|^5 |\RMEvents_i| + |\RMStates|^4)\) is provided in appendix C of \cite{wong1998complexity}.

In words, the first condition guarantees that two states $\RMCommonState_1, \RMCommonState_2 \in \RMStates$ of the RM $\RM$ are members of the same equivalence class if a transition exists between them that is triggered by an event outside of the local event set $\RMEvents_i$. The equivalence classes thus represent the collections of states of $\RM$ that are indistinguishable to an agent who may only observe events from $\RMEvents_i$. The second condition ensures that from any equivalence class, a particular event $\RMCommonEvent \in \RMEvents_i$ may only trigger transitions to a unique successor equivalence class. Using this equivalence relation, we define the projection of $\RM$ onto $\RMEvents_i$. 

\begin{definition}
\label{def:projectionRewardMachine}
(RM projection onto a local event set) Given a reward machine $\RM = (\RMStates, \RMInitialState, \RMEvents, \RMTransition, \RMOutput, \RMFinalStates)$ and a local event set $\RMEvents_i \subseteq \RMEvents$, we define the projection of $\RM$ onto $\RMEvents_i$ as
$\RM_i = (\RMStates_i, \RMInitialState^i, \RMEvents_i, \RMTransition_i, \RMOutput_i, \RMFinalStates_i)$. 

\begin{itemize}
    \item The set of projected states $\RMStates_i$ is given by $\RMStates/\Relation_{i}$; each state $\RMCommonState^i \in \RMStates_i$ is an equivalence class of states from $\RMCommonState \in \RMStates$.
    \item The initial state is $\RMInitialState^i = [\RMInitialState]_{i}$.
    \item Transition function $\RMTransition_i : \RMStates_i \times \RMEvents_i \to \RMStates_i$ is defined such that $\RMCommonState_2^i = \RMTransition_i(\RMCommonState_1^i, \RMCommonEvent)$ if and only if there exist $\RMCommonState_1, \RMCommonState_2\in \RMStates$ such that $\RMCommonState_1^i = [\RMCommonState_1]_{i}$, $\RMCommonState_2^i = [\RMCommonState_2]_{i}$, and $\RMCommonState_2 = \RMTransition(\RMCommonState_1, \RMCommonEvent)$.
    \item The projected set of final states is defined as $\RMFinalStates_i = \{\RMCommonState^i \in \RMStates_i | \exists \RMCommonState \in \RMFinalStates \textrm{ such that } \RMCommonState^i = [\RMCommonState]_{i}\}$.
    \item The output function $\RMOutput_i : \RMStates_i \times \RMStates_i \to \mathbb{R}$ is defined such that $\RMOutput_i(\RMCommonState_1^i, \RMCommonState_2^i) = 1$ if $\RMCommonState_1^i \notin \RMFinalStates_i, \RMCommonState_2^i \in \RMFinalStates_i$ and $\RMOutput(\RMCommonState_1^i, \RMCommonState_2^i) = 0$ otherwise.
\end{itemize}
\end{definition}

The intuition behind this definition is as follows. To define the new set of states, we remove all transitions triggered by events \textit{not} contained in $\RMEvents_i$ and merge the corresponding states. The remaining transitions are used to define the projected transition function $\RMTransition_i$. The reward states $\RMFinalStates_i$ are defined as the collection of merged states containing at least one reward state $\RMCommonState \in \RMFinalStates$ from the original RM, and output function $\RMOutput_i$ is defined accordingly. We note that $\RMTransition_i$ is a well-defined function as a result of condition (2) in the definition of the equivalence relation $\Relation_i$. 

Figures \ref{fig:buttons_local_rm_1}, \ref{fig:buttons_local_rm_2}, and \ref{fig:buttons_local_rm_3} show the results of projecting the task RM from Figure \ref{fig:buttons_rm} onto the local event sets of $\RMEvents_1$, $\RMEvents_2$, and $\RMEvents_3$ respectively. As an example, we specifically examine $\RM_1$  $=$ $\langle \RMStates_1$, $\RMInitialState^1$, $\RMEvents_1$, $\RMTransition_1$, $\RMOutput_1$, $\RMFinalStates_1 \rangle$, illustrated in in Figure \ref{fig:buttons_local_rm_1}. Because events $\greenButton$, $\agentTwoRedButton$, $\agentTwoNotRedButton$, $\agentThreeRedButton$, and $\agentThreeNotRedButton$ are not elements of $\RMEvents_1$, any states connected by these events are merged to form the projected states $\RMStates_1$. For example, states $\RMCommonState_1$, $\RMCommonState_2$, $\RMCommonState_3, \RMCommonState_4$, $\RMCommonState_5 \in \RMStates$ which comprise the diamond structure in Figure \ref{fig:buttons_rm}, are all merged into projected state $\RMCommonState_1^1 \in \RMStates_1$. Intuitively, this portion of the team's RM $\RM$ encodes the necessary coordination between $\buttonsAgentTwo$ and $\buttonsAgentThree$ to press the red button, which is irrelevant to $\buttonsAgentOne$'s portion of the task and is thus represented as a single state in $\RM_1$. Note that $\RM_1$ describes $\buttonsAgentOne$'s contribution to the team's task; press the yellow button then wait for the red button to be pressed, before proceeding to the goal location. Intuitively, this high-level behavior is correct with respect to the team task, regardless of the behavior $\buttonsAgentTwo$ or $\buttonsAgentThree$.

However, clearly not all RMs and local event sets lead to projections describing behavior compatible with the original task. For example, if $\yellowButton \notin \RMEvents_1$, then $\RM_1$ would instead describe a task in which agent \(\buttonsAgentOne\) is not required to press the yellow button. This conflicts with the original task, in which the yellow button must be pressed before agent \(\buttonsAgentTwo\) can proceed across the yellow region. An important notion to define then, is compatibility between the original task, and the task described by a collection of projected RMs.

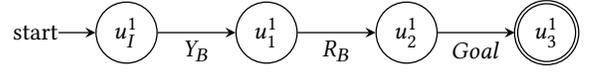
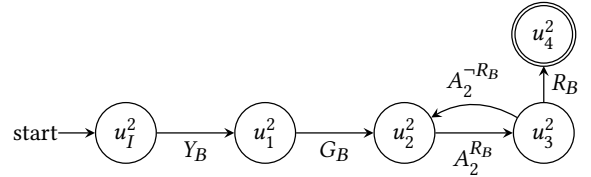
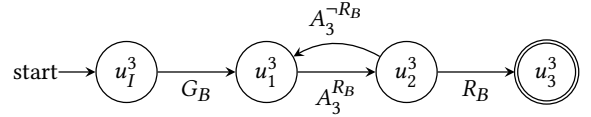
\begin{figure}[t!]
    \centering
    \begin{subfigure}[t]{\columnwidth}
        \centering
        \tikzset{auto,
        ->,
        >=stealth,
        node distance=2.2cm,
        every node/.style={scale=0.85, minimum size=0pt, inner sep=0pt}}

\newcommand{\blockOpacity}{15}
\newcommand{\minBlockHeight}{0.5cm}
\newcommand{\tileDim}{0.5cm}

\usetikzlibrary{shapes.geometric}

\newcommand{\hDist}{1.2cm}
\newcommand{\vDist}{0.75cm}

\resizebox{0.9\textwidth}{!}{

\begin{tikzpicture}

    \node[state, initial] (ui) {$\RMInitialState^1$};
    \path (ui.0)+(\hDist, 0.0cm) node (u1) [state] {$\RMCommonState_1^1$};
    \path (u1.0)+(\hDist, 0.0cm) node (u2) [state] {$\RMCommonState_2^1$};
    \path (u2.0)+(\hDist, 0.0cm) node (u3) [state, accepting] {$\RMCommonState_3^1$};
    
    \draw (ui) -> node[below, yshift=-0.1cm] {$\yellowButton$} (u1);
    \draw (u1) edge node[below, yshift=-0.1cm] {$\redButton$} (u2);
    \draw (u2) edge node[below, yshift=-0.1cm] {$\goal$} (u3);

\end{tikzpicture}}
        \caption{RM projection onto local event set $\RMEvents_1 = \{\yellowButton, \redButton, \goal\}$.}
        \label{fig:buttons_local_rm_1}
    \end{subfigure}%
    ~
    \\
    \vspace{0.3cm}
    \begin{subfigure}[t]{\columnwidth}
        \centering
        \tikzset{auto,
        ->,
        >=stealth,
        node distance=2.2cm,
        every node/.style={scale=0.85, minimum size=0pt, inner sep=0pt}}

\newcommand{\blockOpacity}{15}
\newcommand{\minBlockHeight}{0.5cm}
\newcommand{\tileDim}{0.5cm}

\usetikzlibrary{shapes.geometric}

\newcommand{\hDist}{1.2cm}
\newcommand{\vDist}{0.75cm}

\resizebox{0.9\textwidth}{!}{

\begin{tikzpicture}

    \node[state, initial] (ui) {$\RMInitialState^2$};
    \path (ui.0)+(\hDist, 0.0cm) node (u1) [state] {$\RMCommonState_1^2$};
    \path (u1.0)+(\hDist, 0.0cm) node (u2) [state] {$\RMCommonState_2^2$};
    \path (u2.0)+(\hDist, 0.0cm) node (u3) [state] {$\RMCommonState_3^2$};
    \path (u3.90)+(0.0cm, \vDist) node (u4) [state, accepting] {$\RMCommonState_4^2$};
    
    \draw (ui) -> node[below, yshift=-0.1cm] {$\yellowButton$} (u1);
    \draw (u1) edge node[below, yshift=-0.1cm] {$\greenButton$} (u2);
    \draw (u3) edge node[right, xshift=0.1cm] {$\redButton$} (u4);
    
    \draw (u2) edge node[below, yshift=-0.1cm] {$\agentTwoRedButton$} (u3);
    \draw (u3) edge[bend right] node[above, yshift=0.1cm] {$\agentTwoNotRedButton$} (u2);

\end{tikzpicture}}
        \caption{RM projection onto local event set $\RMEvents_2 = \{\yellowButton, \greenButton, \agentTwoRedButton, \agentTwoNotRedButton, \redButton\}$.}
        \label{fig:buttons_local_rm_2}
    \end{subfigure}
    ~
    \\
    \vspace{0.3cm}
    \begin{subfigure}[t]{\columnwidth}
        \centering
        \tikzset{auto,
        ->,
        >=stealth,
        node distance=2.2cm,
        every node/.style={scale=0.85, minimum size=0pt, inner sep=0pt}}

\newcommand{\blockOpacity}{15}
\newcommand{\minBlockHeight}{0.5cm}
\newcommand{\tileDim}{0.5cm}

\usetikzlibrary{shapes.geometric}

\newcommand{\hDist}{1.2cm}
\newcommand{\vDist}{0.75cm}

\resizebox{0.9\textwidth}{!}{

\begin{tikzpicture}

    \node[state, initial left] (ui) {$\RMInitialState^3$};
    \path (ui.0)+(\hDist, 0.0cm) node (u1) [state] {$\RMCommonState_1^3$};
    \path (u1.0)+(\hDist, 0.0cm) node (u2) [state] {$\RMCommonState_2^3$};
    \path (u2.0)+(\hDist, 0.0cm) node (u3) [state, accepting] {$\RMCommonState_3^3$};
    
    \draw (ui) -> node[below, yshift=-0.1cm] {$\greenButton$} (u1);
    \draw (u2) edge node[below, yshift=-0.1cm] {$\redButton$} (u3);
    
    \draw (u1) edge node[below, yshift=-0.1cm] {$\agentThreeRedButton$} (u2);
    
    \draw (u2) edge[bend right] node[above, yshift=0.1cm] {$\agentThreeNotRedButton$} (u1);

\end{tikzpicture}}
        \caption{RM projection onto local event set $\RMEvents_3 = \{\greenButton, \agentThreeRedButton, \agentThreeNotRedButton, \redButton\}$.}
        \label{fig:buttons_local_rm_3}
    \end{subfigure}
    \caption{Projections of the team RM illustrated in Figure \ref{fig:buttons_task_description}.}
    \label{fig:buttons_projected_rms}
\end{figure}

Consider some finite event sequence $\RMCommonEvent_0 ... \RMCommonEvent_k \in \RMEvents^*$. The natural projection \cite{lin1988observability} of the sequence onto onto $\RMEvents_i^*$, denoted $\Projection_i(\RMCommonEvent_0 ... \RMCommonEvent_k)\in \RMEvents_i^*$, is defined recursively by the relationships $\Projection_i(\EmptyString) = \EmptyString$, $\Projection_i(\RMEventSequence \RMCommonEvent) = \Projection_i(\RMEventSequence) \RMCommonEvent$ if $\RMCommonEvent \in \RMEvents_i$, and $\Projection_i(\RMEventSequence \RMCommonEvent) = \Projection_i(\RMEventSequence)$ if $\RMCommonEvent \notin \RMEvents_i$ for any $\RMEventSequence \in \RMEvents^*$. We remark that $\Projection_i(\RMEventSequence)$ may be thought of as the event sequence $\RMEventSequence$ from the point of view of the $i^{th}$ agent.

Theorem \ref{thm:rm_decomposability} defines a condition guaranteeing that the composition of the individual behaviors described by the projected RMs is equivalent to the behavior described by the original team RM. This condition uses the notions of parallel composition and bisimilarity.

Intuitively, the parallel composition of two or more RMs is a new RM describing all the possible interleavings of their events. If each RM encodes a sub-task, their parallel composition encodes all possible results of carrying out those sub-tasks concurrently. The bisimilarity of two RMs ensures that every sequence of transitions and rewards from one RM can be matched by the other, and vice versa; it provides a formal notion of equivalence between the tasks they encode. These are common concepts for finite transition systems \cite{cassandras2009introduction, baier2008principles, modelingAndControlofLogicalDES_kumar}, that are formally defined for RMs in the supplementary material.


\begin{theorem}
\label{thm:rm_decomposability}
Given RM $\RM$ and a collection of local event sets $\RMEvents_1$, $\RMEvents_2$, ..., $\RMEvents_N$ such that $\bigcup_{i=1}^N\RMEvents_i = \RMEvents$, let $\RM_1$, $\RM_2$, ..., $\RM_N$ be the corresponding collection of projected RMs. Suppose $\RM$ is bisimilar to the parallel composition of $\RM_1$, $\RM_2$, ..., $\RM_N$. Then given an event sequence $\RMEventSequence \in \RMEvents^*$, $\RM(\RMEventSequence) = 1$ if and only if $\RM_i(\Projection_i(\RMEventSequence)) = 1$ for all $i = 1,2, ... , N$. Otherwise, $\RM(\RMEventSequence) = 0$ and $\RM_i(\Projection_i(\RMEventSequence)) = 0$ for all $i=1,2,...,N$. 
\end{theorem}

\begin{proof}
Let $\RM_p = \langle \RMStates_p, \RMInitialState^p, \RMEvents, \RMTransition_p, \RMOutput_p, \RMFinalStates_p \rangle$ be the parallel composition of $\RM_1, ..., \RM_N$. Note that $\RM(\RMEventSequence) = 1$ if and only if $\RMTransition(\RMInitialState, \RMEventSequence) \in \RMFinalStates$. Similarly, $\RM_i(\Projection_i(\RMEventSequence)) = 1$ if and only if $\RMTransition_i(\RMInitialState^i, \Projection_i(\RMEventSequence)) \in \RMFinalStates_i$ for every $i=1,...,N$.  So, it is sufficient to show that $\RMTransition(\RMInitialState, \RMEventSequence) \in \RMFinalStates$ if and only if $\RMTransition_i(\RMInitialState^i, \Projection_i(\RMEventSequence)) \in \RMFinalStates_i$ for every $i=1,...,N$. Given the assumption that $\RM$ is bisimilar to $\RM_p$, we can show by induction that $\RMTransition(\RMInitialState, \RMEventSequence) \in \RMFinalStates$ if and only if $\RMTransition_p(\RMInitialState^p, \RMEventSequence) \in \RMFinalStates_p$ (see chapter 7 of \cite{baier2008principles}). Now, given the definition of parallel composition, it is readily seen that $\RMTransition_p(\RMInitialState^p, \RMEventSequence) \in \RMFinalStates_p$ if and only if $\RMTransition_i(\RMInitialState^i, \Projection_i(\RMEventSequence)) \in \RMFinalStates_i$ for every $i = 1,...,N$ \cite{modelingAndControlofLogicalDES_kumar}.
\end{proof}

We note that \cite{karimadini2011guaranteed, karimadini2016cooperative} present conditions, in terms of $\RM$ and $\RMEvents_1$, ..., $\RMEvents_N$, which may be applied to check whether $\RM$ is bisimilar to the parallel composition of its projections $\RM_1$, ..., $\RM_N$. Alternatively, one may computationally check whether this result holds by automatically constructing the parallel composition of the projected RMs and applying the \textit{Hopcroft-Karp} algorithm to check bisimilarity \cite{hopcroftKarp1971linear, bonchi2013checking}. The runtime of this algorithm is \(\mathcal{O}(|\RMEvents|(|\RMStates|+|\RMStates_p|))\), where \(\RMStates\) are the states of \(\RM\), and \(\RMStates_p\) are the states of the parallel composition of \(\RM_1, ..., \RM_N\) \cite{hopcroftKarp1971linear}. If the bisimilarity condition does not hold, the task designer might add to the local event sets, giving the agents access to more information, and re-check the condition.


\section{Decentralized Q-Learning with Projected Reward Machines (DQPRM)}

Inspired by Theorem \ref{thm:rm_decomposability}, we propose a distributed approach to learning a decentralized policy. Our idea is to use the projected RMs to define a collection of single-agent RL tasks, and to train each agent on their respective task using the QRM algorithm described in \S \ref{sec:labeling_functions}.

For clarity, we wish to train the agents using their projected RMs in an \textit{individual setting}: the agents take actions in the environment in the absence of their teammates. However, the policies they learn should result in a team policy that is successful in the \textit{team setting}, in which the agents interact simultaneously with the shared environment.

\subsection{Local Labeling Functions and Shared Event Synchronization}

\label{sec:local_labeling_functions}

Projected reward machine $\RM_i$ defines the task of the $i^{th}$ agent in terms of high-level events from $\RMEvents_i$. As discussed in \S \ref{sec:labeling_functions}, to connect $\RM_i$ with the underlying environment, a labeling function is required to define the environment states that cause the events in $\RMEvents_i$ to occur. In the team setting, we can intuitively define a labeling function $\RMLabelingFunction : \SGJointStates \times \RMStates \to 2^{\RMEvents}$ mapping team states $\SGCommonJointState \in \SGJointStates$ and RM states $\RMCommonState \in \RMStates$ to sets of events. For example, $\RMLabelingFunction(\SGCommonJointState, \RMCommonState) = \{\redButton\}$ if $\SGCommonJointState$ is such that $\buttonsAgentTwo$ and $\buttonsAgentThree$ are pressing the red button and $\RMCommonState = \RMCommonState_5$. 

To use $\RM_i$ in the individual setting however, we must first define a \textit{local labeling function} $\RMLabelingFunction_i : \SGStates_i \times \RMStates_i \to 2^{\RMEvents_i}$ mapping the local states $\SGCommonState_i \in \SGStates_i$ and projected RM states $\RMCommonState^i \in \RMStates_i$ to sets of events in $\RMEvents_i$. Operating under the assumption that only one event can occur per agent per time step, we require that, for any local state pair $(\SGCommonState_i, \RMCommonState^i)$, $\RMLabelingFunction_i(\SGCommonState_i, \RMCommonState^i)$ returns at most a single event from $\RMEvents_i$. Furthermore, the local labeling functions $\RMLabelingFunction_1$, $\RMLabelingFunction_2$,..., $\RMLabelingFunction_N$ should be defined such that they always collectively output the same set of events as $\RMLabelingFunction$, when being used in the team setting.

For a given event $\RMCommonEvent \in \RMEvents$, we define the set $I_{\RMCommonEvent} = \{i | \RMCommonEvent \in \RMEvents_i\}$ as the \textit{collaborating agents on $\RMCommonEvent$}. 

\begin{definition} (Decomposable labeling function)
\label{def:labeling_function_decomposability}
A labeling function $\RMLabelingFunction : \SGJointStates \times \RMStates \to 2^{\RMEvents}$ is considered decomposable with respect to local event sets $\RMEvents_1$, $\RMEvents_2$, ..., $\RMEvents_N$ if there exists a collection of local labeling functions $\RMLabelingFunction_1$, $\RMLabelingFunction_2$, ..., $\RMLabelingFunction_N$ with $\RMLabelingFunction_i : \SGStates_i \times \RMStates_i \to 2^{\RMEvents_i}$ such that:

\begin{enumerate}
    \item $|\RMLabelingFunction_i(\SGCommonState_i, \RMCommonState^i)| \leq 1$ for every $\SGCommonState_i \in \SGStates_i$ and every $\RMCommonState^i \in \RMStates_i$.
    \item $\RMLabelingFunction(\SGCommonState, \RMCommonState)$ outputs event $\RMCommonEvent$ if and only if $\RMLabelingFunction_i(\SGCommonState_i, \RMCommonState^i)$ outputs event $\RMCommonEvent$ for every $i$ in $I_\RMCommonEvent$. Here, $\SGCommonState_i$ is the $i^{th}$ component of the team's joint state $\SGCommonJointState$, and $\RMCommonState^i \in \RMStates_i$ is the state of RM $\RM_i$ containing state $\RMCommonState \in \RMStates$ from RM $\RM$ (recall that states in $\RMStates_i$ correspond to collections of states from $\RMStates$).
\end{enumerate}

\end{definition}

Note that $\RMLabelingFunction$ will be decomposable if we can define $\RMLabelingFunction_1$, ..., $\RMLabelingFunction_N$ to satisfy the conditions in Definition \ref{def:labeling_function_decomposability}. Following this idea, we conceptually construct $\RMLabelingFunction_i$ from $\RMLabelingFunction$ as follows:  $\RMLabelingFunction_i(\bar{\SGCommonState}_i, \RMCommonState^i)$ outputs event $\RMCommonEvent \in \RMEvents_i$ whenever there exists a possible configuration of agent $i$'s teammates $\SGCommonJointState = (\SGCommonState_1, ..., \bar{\SGCommonState}_i, ... , \SGCommonState_N)$ such that $\RMLabelingFunction(\SGCommonJointState, \RMCommonState)$ outputs $\RMCommonEvent$, where $\RMCommonState \in \RMStates$ is any state belonging to $\RMCommonState^i\in \RMStates_i$. Our interpretation of this definition of $\RMLabelingFunction_i$ is as follows. While $\RMLabelingFunction(\SGCommonJointState, \RMCommonState)$ outputs the events that occur when the team is in joint state $\SGCommonJointState$ and RM state $\RMCommonState$, the local labeling function $\RMLabelingFunction_i : \SGStates_i \times \SGActions_i \to 2^{\RMEvents_i}$ outputs the events in $\RMEvents_i$ that \textit{could} be occurring from the point of view of an agent who knows $\RMLabelingFunction$, but may only observe $\SGCommonState_i \in \SGStates_i$ and $\RMCommonState^i \in \RMStates_i$. 

Furthermore, to ensure $\RMLabelingFunction_1,...,\RMLabelingFunction_N$ output an event $\RMCommonEvent$ only if $\RMLabelingFunction$ does, we also require that each event $\RMCommonEvent \in \RMEvents$ be "under the control of at least one of the agents" in the following sense: if the agent is not in some particular subset of local states, the event will not be returned by $\RMLabelingFunction$, regardless of the states of the agent's teammates. A more formal definition of this construction of $\RMLabelingFunction_i$ as well as conditions on $\RMLabelingFunction$ that ensure $\RMLabelingFunction_i$ are well defined, are given in the supplementary material. 

We say event $\RMCommonEvent \in \RMEvents$ is a \textit{shared event} if it belongs to the local event sets of multiple agents, i.e., if $|I_{\RMCommonEvent}| > 1$. In the buttons task, $\yellowButton \in \RMEvents_1 \cap \RMEvents_2$ is an example of a shared event. Suppose $\buttonsAgentOne$ and $\buttonsAgentTwo$ use the events output by $\RMLabelingFunction_1$ and $\RMLabelingFunction_2$, respectively, to update $\RM_1$ and $\RM_2$, while interacting in the team setting. Because $\RMEvents_1$ and $\RMEvents_2$ both include the event $\yellowButton$, the agents must \textit{syncrhonize} on this event: $\yellowButton$ should simultaneously cause transitions in both projected RMs, or it should cause a transition in neither of them. In practice, synchronization on shared events is implemented as follows: If $\RMLabelingFunction_i$ returns a shared event $\RMCommonEvent$, the $i^{th}$ agent should check with all teammates in $I_{\RMCommonEvent}$ whether their local labeling functions also returned $\RMCommonEvent$, before using the event to update $\RM_i$. Event synchronization corresponds to the agents communicating and collectively acknowledging that the shared event occurred, before progressing through their repsective tasks. 

Given a sequence of joint states $\SGCommonJointState_0 \SGCommonJointState_1 ... \SGCommonJointState_k$ in the team setting, we may use $\RMLabelingFunction$ and $\RM$ to uniquely define a corresponding sequence $\RMLabelingFunction(\SGCommonJointState_o ... \SGCommonJointState_k) \in \RMEvents^*$ of events. Similarly, given the corresponding collection of sequences $\{s_0^i...s_k^i\}_{i=1}^N$ of local states, local labeling functions $\RMLabelingFunction_1$, ..., $\RMLabelingFunction_N$, and assuming that the agents synchronize on shared events, we may define the corresponding sequences $\RMLabelingFunction_i(\SGCommonState_0^i...\SGCommonState_k^i) \in \RMEvents_i^*$ of events for every $i = 1,2,...,N$. A step-by-step construction of sequences $\RMLabelingFunction(\SGCommonJointState_0...\SGCommonJointState_k)$ and $\RMLabelingFunction_i(\SGCommonState_0^i...\SGCommonState_k^i)$, as well as the details of the induction step of the proof of Theorem \ref{thm:reward_equivalence} are provided in the supplementary material.

\begin{theorem}
\label{thm:reward_equivalence}

Given $\RM$, $\RMLabelingFunction$, and $\RMEvents_1$, ..., $\RMEvents_N$, suppose the bisimilarity condition from Theorem \ref{thm:rm_decomposability} holds. Furthermore, assume $\RMLabelingFunction$ is decomposable with respect to $\RMEvents_1$, ..., $\RMEvents_N$ with the corresponding local labeling functions $\RMLabelingFunction_1$, ..., $\RMLabelingFunction_N$. Let $\SGCommonJointState_0...\SGCommonJointState_k$ be a sequence of joint environment states and $\{s_0^i...s_k^i\}_{i=1}^N$ be the corresponding sequences of local states. If the agents synchronize on shared events, then $\RM(\RMLabelingFunction(\SGCommonJointState_0...\SGCommonJointState_k)) = 1$ if and only if $\RM_i(\RMLabelingFunction_i(\SGCommonState_0^i...\SGCommonState_k^i)) = 1$ for all $i = 1,2,...,N$. Otherwise $\RM(\RMLabelingFunction(\SGCommonJointState_0...\SGCommonJointState_k)) = 0$ and $\RM_i(\RMLabelingFunction_i(\SGCommonState_0^i...\SGCommonState_k^i)) = 0$ for all $i = 1,2,...,N$.

\end{theorem}

\begin{proof}
Note that given the result of Theorem \ref{thm:rm_decomposability}, it is sufficient to show that for every $i=1,2,...,N$, the relationship  $P_{i}(\RMLabelingFunction(\SGCommonJointState_0 \SGCommonJointState_1 ... \SGCommonJointState_k)) = \RMLabelingFunction_i(\SGCommonState_0^i \SGCommonState_1^i ... \SGCommonState_k^i)$ holds. In words, we wish to show that for every $i$, the projection of the sequence output by labeling function $\RMLabelingFunction$ is equivalent to the sequence of synchronized outputs of local labeling function $\RMLabelingFunction_i$. To do this, we show that $l_t \cap \RMEvents_i = \Tilde{l}_t^i$ for every $i = 1,...,N$ and every $t = 1, ... ,k$. Here, $l_t$ denotes the output of labeling function $\RMLabelingFunction$ at time $t$, and $\Tilde{l}_t^i$ denotes the synchronized output of local labeling function $\RMLabelingFunction_i$ at time $t$. 

At time $t=0$, $l_0, \Tilde{l}_0^i$ are defined to be the empty set and the above condition holds trivially. We also have that at time $t = 0$, $\RMInitialState^i \in \RMInitialState$ by definition of initial state $\RMInitialState^i$. If, for any time $t$, $\RMCommonState_t \in \RMCommonState_t^i$ for each $i = 1,2,...,N$, we can use the definition of the decomposability of labeling function $\RMLabelingFunction$ to show that $l_{t+1}\cap \RMEvents_i = \Tilde{l}_{t+1}^i$ for each $i$. Using the definition of the projected RMs, we may show that $l_{t+1}\cap \RMEvents_i = \Tilde{l}_{t+1}^i$ implies $\RMCommonState_{t+1} \in \RMCommonState_{t+1}^i$. Thus by induction, we conclude the proof.

\end{proof}

Let $V^{\JointPolicy}(\SGCommonJointState)$ denote the expected sum of future undiscounted rewards returned by $\RM$, given the team follows joint policy $\JointPolicy$ from joint environment state $\SGCommonJointState \in \SGJointStates$ and initial RM state $\RMInitialState$. Similarly, let $V_i^{\JointPolicy}(\SGCommonJointState)$ denote the expected future reward returned by $\RM_i$, given the team follows the same policy from state $\SGCommonJointState$ and projected initial RM state $\RMInitialState^i$. Using the result of Theorem \ref{thm:reward_equivalence} and the Fr\'echet conjunction inequality, we provide the following upper and lower bounds on the value function corresponding to $\RM$, in terms of the value functions corresponding to projected RMs $\RM_i$.

\begin{theorem}
\label{thm:frechet_bound}
If the conditions in Theorem \ref{thm:reward_equivalence} are satisfied, then
\begin{multline*}
\max\{0, V_1^{\JointPolicy}(\SGCommonJointState) + V_2^{\JointPolicy}(\SGCommonJointState) + ... + V_N^{\JointPolicy}(\SGCommonJointState)-(N-1)\} \leq V^{\JointPolicy}(\SGCommonJointState) \\ \leq \min\{V_1^{\JointPolicy}(\SGCommonJointState), V_2^{\JointPolicy}(\SGCommonJointState), ..., V_N^{\JointPolicy}(\SGCommonJointState)\}.
\end{multline*}
\end{theorem}

\begin{proof}
Note that because $\RM$ returns 1 if the task it encodes is completed and 0 otherwise, $V^{\JointPolicy}(\SGCommonJointState_0)$ is equivalent to the probability of the task being completed, given the team follows policy $\JointPolicy$ from the initial state $\SGCommonJointState_0$ and initial reward machine state $\RMInitialState$. Similarly, $V_i^{\JointPolicy}(\SGCommonJointState_0)$ is the probability of satisfying the task encoded by $\RM_i$.

By the result of Theorem 2, for any sequence of team states $\SGCommonJointState_0...\SGCommonJointState_k$, $\RM(\RMLabelingFunction(\SGCommonJointState_0...\SGCommonJointState_k)) = 1$ if and only if $\RM_i(\RMLabelingFunction_i(\SGCommonState_1^i...\SGCommonState_k^i)) = 1$ for all $i = 1,...,N$. So, the likelihood of completing the task encoded by $\RM$ under policy $\JointPolicy$ is equivalent to the likelihood of simultaneously completing all the tasks encoded by the collection $\{\RM_i\}_{i=1}^N$ under the same policy. Using our interpretation of $V^{\JointPolicy}(\SGCommonJointState_0)$ and $V^{\JointPolicy}_i(\SGCommonJointState_0)$ as these probabilities, we apply the Fr\'echet conjunction inequality to arrive at the final result.
\end{proof}

\subsection{Training and Evaluating}

\begin{figure}
    \centering
    \begin{subfigure}[t]{\columnwidth}
        \begin{center}
        \includegraphics[width=0.92\textwidth]{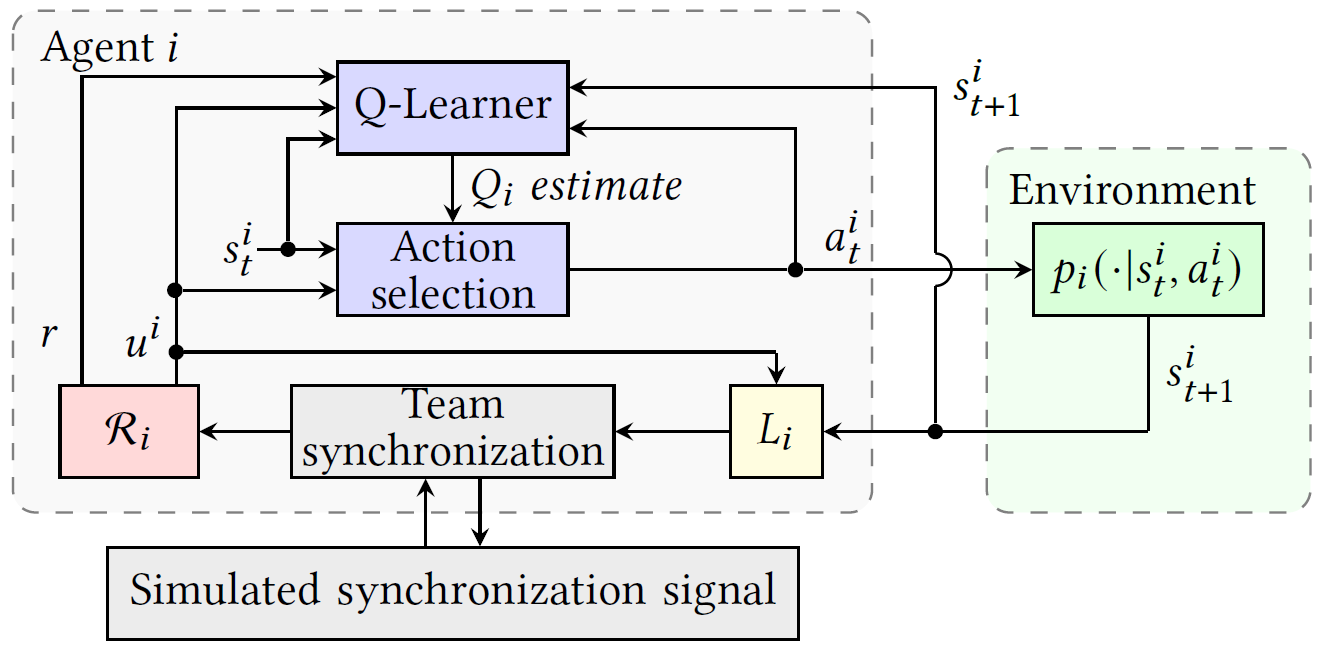}
        \end{center}
        \caption{DQPRM policy training in individual setting.}
        \label{fig:dqprm_training}
    \end{subfigure}
    ~
    \\
    \vspace{0.2cm}
    \begin{subfigure}[t]{\columnwidth}
        \begin{center}
        \includegraphics[width=0.92\textwidth]{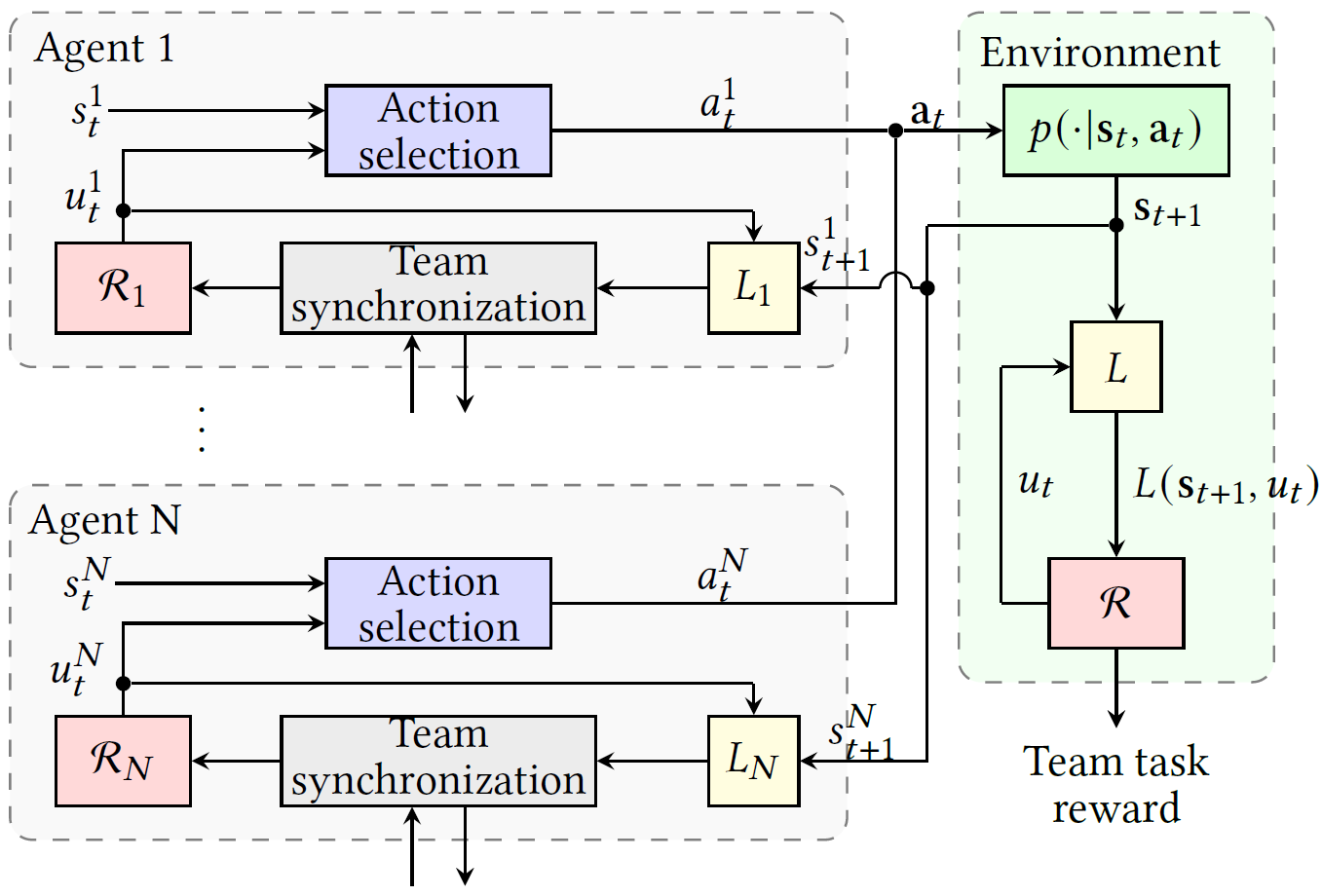}
        \end{center}
        \caption{DQPRM policy execution in team setting.}
        \label{fig:dqprm_execution}
    \end{subfigure}
    \caption{}
    \label{fig:decentralized_testing_fig}
\end{figure}

Theorem \ref{thm:reward_equivalence} tells us that it makes no difference whether we use RM $\RM$ and team labeling function $\RMLabelingFunction$, or projected RMs $\RM_1$,...,$\RM_N$ and local labeling functions $\RMLabelingFunction_1$, ..., $\RMLabelingFunction_N$ to describe the team task. By replacing $\RM$ and $\RMLabelingFunction$ with $\RM_1$, ..., $\RM_N$ and $\RMLabelingFunction_1$, ..., $\RMLabelingFunction_N$ however, we note that the only interactions each agent has with its teammates are synchronizations on shared events. 

This key insight provides a method to train the agents separately from their teammates. We train each agent in an individual setting, isolated from its teammates, using rewards returned from $\RM_i$ and events returned from $\RMLabelingFunction_i$. Whenever $\RMLabelingFunction_i$ outputs what would be a shared event in the team setting, we randomly provide a simulated synchronization signal with a fixed probability of occurrence. Figure \ref{fig:dqprm_training} illustrates this approach. 

By simulating the synchronization on shared events, we take an optimistic approach to decentralized learning. Each agent learns to interact with idealized teammates in the sense that during training, any shared events necessary for task progression will always occur, albeit after some random amount of time.

During training, each agent individually performs q-learning to find an optimal policy for the sub-task described by its projected RM, similarly to as described in \S \ref{sec:labeling_functions}. The $i^{th}$ agent learns a collection of q-functions $Q_i = \{\qValue_{\RMCommonState^i} | \RMCommonState^i \in \RMStates_i\}$ such that each q-function $\qValue_{\RMCommonState^i} : \SGStates_i \times \SGActions_i \rightarrow \mathbb{R}$ corresponds to the agent's optimal policy while it is in projected RM state $\RMCommonState^i$.

To evaluate the learned policies, we test the team by allowing the agents to interact in the team environment and evaluate the team's performance using team task RM $\RM$. Each agent tracks its own task progress using its projected RM $\RM_i$ and follows the policy it learned during training, as shown in Figure \ref{fig:dqprm_execution}.
\section{Experimental Results}
\label{sec:experimental_results}

In this section, we provide empirical evaluations of DQPRM in three task domains.\footnote{Project code is publicly available at: \href{https://github.com/cyrusneary/rm-cooperative-marl}{github.com/cyrusneary/rm-cooperative-marl}.} The buttons task is as described in \S \ref{sec:reward_machines_for_MARL}. We additionally consider two-agent and ten-agent rendezvous tasks in which each agent must simultaneously occupy a specific rendezvous location before individually navigating to separate goal locations. 

We compare DQPRM's performance against three baseline algorithms: the naive centralized QRM (CQRM) algorithm described in \S \ref{sec:labeling_functions}, independent q-learners (IQL) \cite{MARLIndependentvsCooperativeAgentsTan}, and hierarchical independent learners (h-IL) \cite{hierarchicalDeepMARL}. Because of the non-Markovian nature of the tasks, we provide both the IQL and h-IL agents with additional memory states. In the buttons task, the memory states encode which buttons have already been pressed. In the rendezvous task, the memory state encodes whether or not the team has completed the rendezvous, and whether each agent has reached its goal.

Each IQL agent learns a q-function mapping augmented state-action pairs to values. That is, the $i^{th}$ agent learns a q-function $\qValue_i : \SGStates_i \times \SGStates_{M_i} \times \SGActions_i \to \mathbb{R}$, where $\SGStates_i$, $\SGActions_i$ are the local states and actions of the agent and $\SGStates_{M_i}$ is the finite set of its memory states. 

Our implementation of h-IL is inspired by the learning structure outlined in \cite{hierarchicalDeepMARL}. Each agent uses tabular q-learning to learn a meta-policy --- which uses the current memory state to select a high-level option --- as well as a collection of low-level policies --- which implement those options in the environment.  The available options correspond to the high-level tasks available to each agent. For example, $\buttonsAgentOne$ in the buttons task is provided with the following three options: remain in a non-colored region, navigate to the yellow button, and navigate to the goal location. Furthermore, we prune the available options when necessary. For example, before the red button has been pressed, \(\buttonsAgentOne\) cannot cross the red region to reach the goal, so, it doesn't have access to the corresponding option.

In all algorithms, we use a discount factor $\SGDiscount = 0.9$ and a learning rate $\learningRate = 0.8$. For action selection, we use softmax exploration with a constant temperature parameter $\tau = 0.02$ \cite{tijsma2016comparing}. In the training of the DQPRM agents, if an agent observes a shared event, then with probability $0.3$, it is provided with a signal simulating successful synchronization with all collaborators. 

All tasks are implemented in a 10x10 gridworld and all agents have the following 5 available actions: move right, move left, move up, move down, or don't move. If an agent moves in any direction, then with a $2\%$ chance the agent will instead slip to an adjacent state. Each episode lasts 1,000 time steps. We perform periodic testing episodes in which the agents exploit the policies they have learned and team's performance is recorded. 

\begin{figure}
    \centering
    \begin{subfigure}[t]{\columnwidth}
        \centering
        \includegraphics[width=0.9\textwidth]{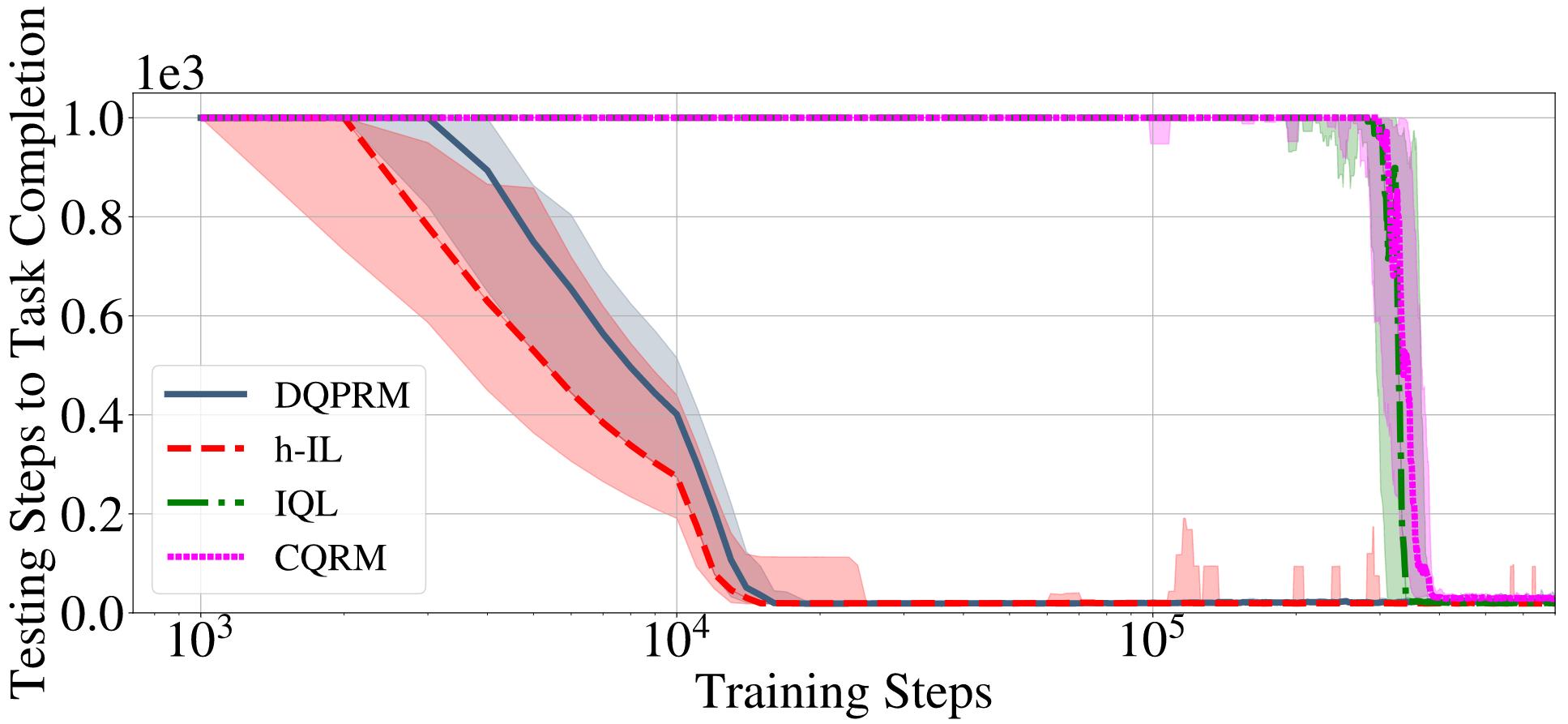}
        \caption{Two-agent rendezvous task.}
        \label{fig:two_agent_rendezvous_results}
    \end{subfigure}
    ~
    \\
    \begin{subfigure}[t]{\columnwidth}
        \centering
        \includegraphics[width=0.9\textwidth]{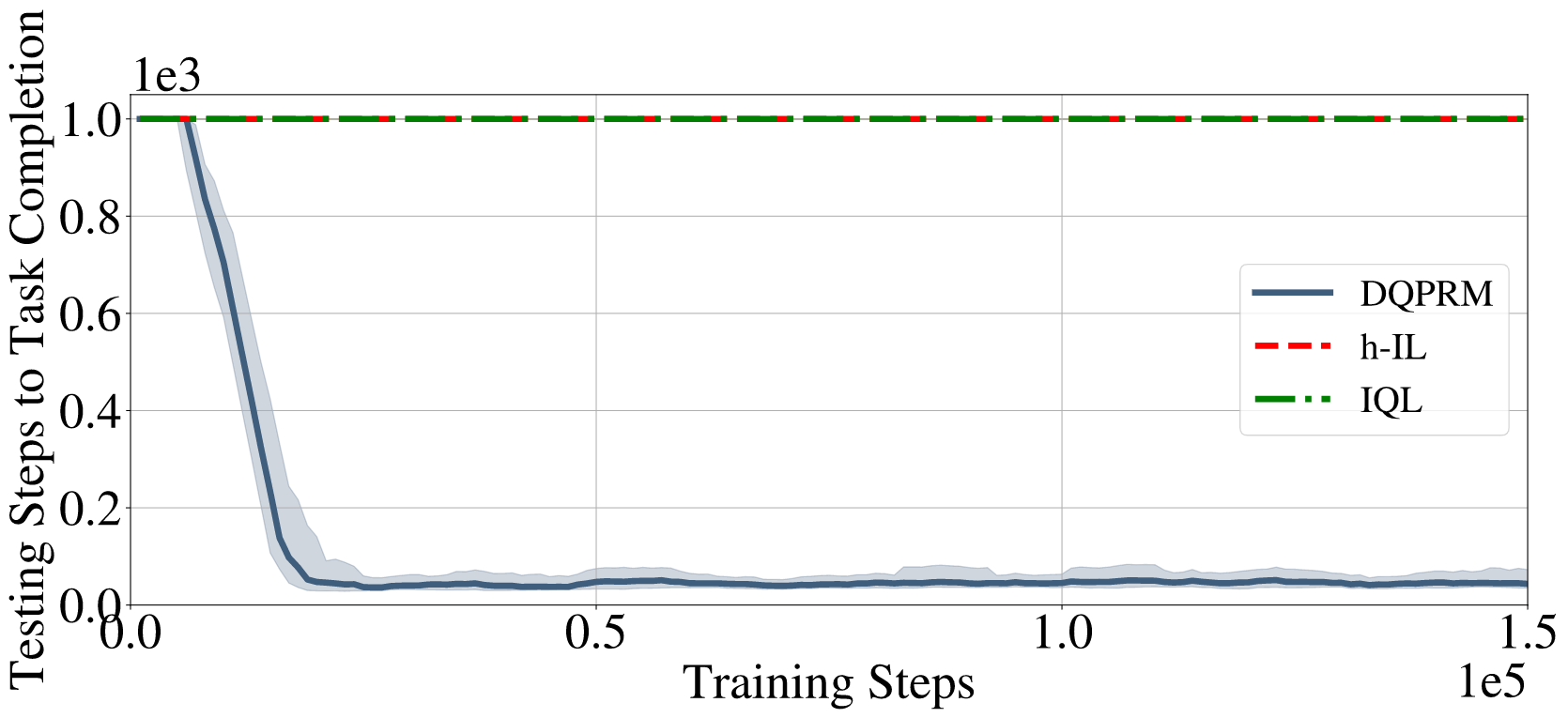}
        \caption{Ten-agent rendezvous task.}
        \label{fig:ten_agent_rendezvous_results}
    \end{subfigure}
    ~
    \\
    \begin{subfigure}[t]{\columnwidth}
        \centering
        \includegraphics[width=0.9\textwidth]{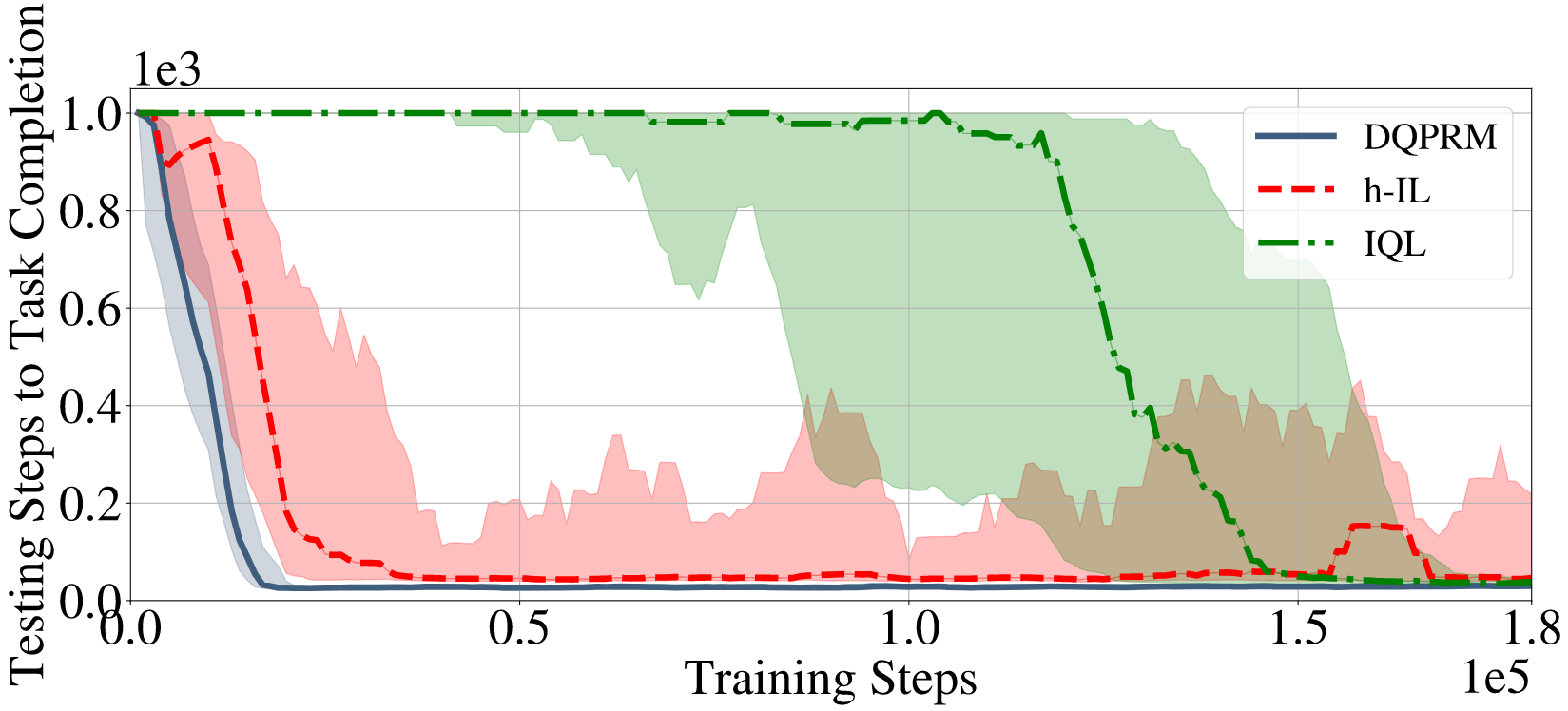}
        \caption{Three-agent buttons task.}
        \label{fig:buttons_results}
    \end{subfigure}
    \caption{Algorithm performance on various tasks. Lower is better. The y-axis show the number of steps required for the learned policies to complete the task. The x-axis shows the number of elapsed training steps.}
    \label{fig:results}
\end{figure}

Figure \ref{fig:results} shows the experimental results for each algorithm over 10 separate runs per domain. The figures plot the median number of testing steps required to complete the team task against the number of elapsed training steps. Because DQPRM trains each agent individually, one training step on the plot refers to one training step taken by each individual agent. The shaded regions enclose the $25^{th}$ and $75^{th}$ percentiles. We note that the CQRM baseline is only tested in the two-agent scenario because the centralized approach requires excessive amounts of memory for more agents; storing a centralized policy for three agents requires approximately two billion separate values.

While the h-IL baseline marginally outperforms the proposed DQPRM algorithm in the two-agent rendezvous task, in the more complex tasks involving more agents, DQPRM outperforms all baseline methods. In the ten-agent rendezvous task, the baseline methods fail to learn a successful team policy within the allowed number of training steps. This demonstrates the ability of the proposed DQPRM algorithm to scale well with the number of agents. In the buttons task, DQPRM quickly learns a policy that completes the task 20 steps faster than that of the h-IL baseline. This difference in performance likely arises because the hierarchical approach has pruned optimal policies, whereas QRM guarantees that each agent converges to the optimal policy for their sub-task \cite{icarte2018using}.

The key advantage of the DQPRM algorithm, is that it uses the information provided in the RM to train the agents entirely separately. This removes the problem of non-stationarity and it allows each agent to more frequently receive reward during training. This is especially beneficial to the types of tasks we study, which have sparse and delayed feedback. We further discuss the differences between DQPRM and hierarchical approaches to MARL in \S \ref{sec:related_work}.
\section{Related Work}
\label{sec:related_work}

Task decomposition in multi-agent systems has been studied from a planning and cooperative control perspective \cite{karimadini2011guaranteed, karimadini2016cooperative, decentralized_uav_control,dai2014automatic}. These works examine the conditions in which group tasks described by automata may be broken into sub-tasks executable by individuals. \cite{Djeumou2020, cubuktepe2020policy} provide methods to synthesize control policies for large-scale multi-agent systems with temporal logic specifications. However, all of these works assume a known model of the environment, differing from the learning setting of this paper.

Several works have explored combining formal methods and MARL. Recently, the authors of \cite{leon2020extended} present Extended Markov Games; a mathematical model allowing multiple agents to concurrently learn to satisfy multiple non-Markovian task specifications. The authors of \cite{muniraj2018enforcing} use minimax deep q-learning to solve zero-sum adversarial games in which the reward is encoded by signal temporal logic specifications. However, to the best of our knowledge, none have yet studied how automata-based task descriptions can be used to decompose a cooperative problem in a way that allows the agents to learn in the absence of their teammates.

The MARL literature is rich \cite{zhang2019multi, non_stationary_MARL_survey, hernandez2019survey, IQLSurveyMatingon}. A popular popular approach is IQL \cite{MARLIndependentvsCooperativeAgentsTan}; each agent learns independently and treats its teammates as part of the environment. \cite{guestrin2002coordinated, kok2005using, van2016coordinated} decompose cooperative tasks by factoring the joint q-function into components. 

More recently, centralized training decentralized execution (CTDE) paradigm algorithms, such as QMIX \cite{rashid2018qmix}, have shown empirical success in cooperative deep MARL problems \cite{sunehag2018value, rashid2018qmix, son2019qtran, mahajan2019maven}. These algorithms enforce the assumption that the team's q-function can be decomposed in a way that allows the agents to make decentralized decisions. CTDE algorithms have been shown to perform well on challenging tasks \cite{samvelyan2019starcraft}. However, the centralized training required by these methods can be sample inefficient, as observed in the results of our centralized approach to learning with reward machines, shown in Figure \ref{fig:two_agent_rendezvous_results}. Our work studies how take advantage of tasks in which only sparse interactions are required between the agents, to avoid simultaneous training of the agents altogether. 

Our work, which examines cooperative tasks that have sparse and temporally delayed rewards, is most closely related to hierarchical approaches to MARL. In particular, \cite{HMARL2001, HMARL2006} use task hierarchies to decompose the multi-agent problem. By learning cooperative strategies only in terms of the sub-tasks at the highest levels of the hierarchy, agents learn to coordinate much more efficiently than if they were sharing information at the level of primitive state-action pairs. More recently, \cite{hierarchicalDeepMARL} empirically demonstrates the effectiveness of a \textit{deep} hierarchical approach to certain cooperative MARL tasks. A key difference between task hierarchies and RMs, is that RMs explicitly encode the temporal ordering of the high-level sub-tasks. It is by taking advantage of this information that we are able to break a team's task into components, and to train the agents independently while guaranteeing that they are learning behavior appropriate for the original problem. Conversely, in a hierarchical approach, the agents must still learn to coordinate at the level of sub-tasks. Thus, the learning problem remains inherently multi-agent, albeit simplified.

In this work, we assume the task RM is known, and present a method to use its decomposition to efficiently solve the MARL problem. The authors of \cite{xu2019joint, icarte2019learning} demonstrate that, in the single-agent setting, RMs can be learned from experience, removing the assumption of the RM being known \textit{a priori} by the learner. This presents an interesting direction for future research: how may agents learn, in a multi-agent setting, RMs encoding either the team's task or projected local tasks. Furthermore, \cite{icarte2018using, rm_journal_version} demonstrate in the single-agent setting that RMs may be applied to continuous environments by replacing tabular q-learning with double deep q-networks \cite{doubldDeepQLearning}. This extension to more complex environments also readily applies to our work, which decomposes multi-agent problems into collections of RMs describing single-agent tasks. 
\section{Conclusions}

In this work, we propose a reward machine (RM) based task representation for cooperative multi-agent reinforcement learning (MARL). The representation allows for a team's task to be decomposed into sub-tasks for individual agents. We accordingly propose a decentralized q-learning algorithm that effectively reduces the MARL problem to a collection of single-agent problems. Experimental results demonstrate the efficiency and scalability of the proposed algorithm, which learns successful team policies even when the baseline algorithms do not converge within the allowed training period. This work demonstrates how well-suited RMs are to the specification and decomposition of MARL problems, and opens interesting directions for future research.

\begin{acks}
This work was supported in part by ARO W911NF-20-1-0140, DARPA D19AP00004, and ONR N00014-18-1-2829.
\end{acks}



\bibliographystyle{ACM-Reference-Format} 
\bibliography{bibliography}

\pagebreak

\newgeometry{left=1in, right=1in}
\onecolumn
\begin{center}
\huge \textbf{Reward Machines for Cooperative Multi-Agent Reinforcement Learning}\\
\LARGE Supplementary Material
\end{center}
\section{Parallel Composition and Bisimulation}

\begin{definition}
\label{def:parallel_composition}
(Parallel composition of RMs) The \textit{parallel composition} of two reward machines $\RM_i = \langle \RMStates_i, \RMInitialState^i, \RMEvents_i, \RMTransition_i, \RMOutput_i, \RMFinalStates_i \rangle$, $i = 1,2$ is defined as $\RM_1 \parallel \RM_2 = \langle \RMStates, \RMInitialState, \RMEvents, \RMTransition, \RMOutput, \RMFinalStates \rangle$ where

\begin{itemize}
    \item The set of states is defined as $\RMStates = \RMStates_1 \times \RMStates_2$.
    \item The initial state is $\RMInitialState = (\RMInitialState^1, \RMInitialState^2)$.
    \item The set of events is defined as $\RMEvents = \RMEvents_1 \cup \RMEvents_2$.
    \item $\RMTransition$ is defined such that for every $(\RMCommonState_1, \RMCommonState_2) \in \RMStates_1 \times \RMStates_2$ and for every event $\RMCommonEvent \in \RMEvents$, 
    \[\begin{aligned}
        \RMTransition((\RMCommonState_1, \RMCommonState_2), \RMCommonEvent) = \begin{cases} 
          (\RMTransition_1(\RMCommonState_1, \RMCommonEvent), \RMTransition(\RMCommonState_2, \RMCommonEvent)), & \textrm{if } \; \RMTransition_1(\RMCommonState_1, \RMCommonEvent), \; \RMTransition_2(\RMCommonState_2, \RMCommonEvent) \textrm{ defined}, \; \RMCommonEvent \in \RMEvents_1 \cap \RMEvents_2 \\
          
          (\RMTransition_1(\RMCommonState_1, \RMCommonEvent), \RMCommonState_2), & \textrm{if } \; \RMTransition_1(\RMCommonState_1, \RMCommonEvent) \textrm{ defined}, \; \RMCommonEvent \in \RMEvents_1 \setminus \RMEvents_2\\
          
          (\RMCommonState_1, \RMTransition_2(\RMCommonState_2, \RMCommonEvent)), & \textrm{if } \; \RMTransition_2(\RMCommonState_2, \RMCommonEvent) \textrm{ defined}, \; \RMCommonEvent \in \RMEvents_2 \setminus \RMEvents_1\\ 
          
          \textrm{undefined}, & \textrm{otherwise}
       \end{cases}
    \end{aligned}\]
    
    \item The set of final states is defined as $\RMFinalStates = \RMFinalStates_1 \times \RMFinalStates_2$
    
    \item The output function $\RMOutput : \RMStates \times \RMStates \to \mathbb{R}$ is defined such that $\RMOutput(\RMCommonState, \RMCommonState') = 1$ if $\RMCommonState_1 \notin \RMFinalStates$, $\RMCommonState_2 \in \RMFinalStates$ and $\RMOutput = 0$ otherwise.
\end{itemize}

By $\parallel_{i=1}^N\RM_i$ we denote the parallel composition of a collection of RMs $\RM_1$, $\RM_2$, ..., $\RM_N$. The parallel composition of more than two RMs is defined using the associative property of the parallel composition operator $\parallel_{i=1}^N \RM_i = \RM_1 \parallel(\RM_2 \parallel(... \parallel (\RM_{N-1} \parallel\RM_N)))$ \cite{cassandras2009introduction}.

\end{definition}

\begin{definition}
(Bisimilarity of Reward Machines) Let $\RM_i = \langle \RMStates_i, \RMInitialState^i, \RMEvents, \RMTransition_i, \RMOutput_i, \RMFinalStates_i \rangle$, $i = 1, 2$, be two RMs. $\RM_1$ and $\RM_2$ are bisimilar, denoted $\RM_1 \cong \RM_2$ if there exists a relation $\Relation \subseteq \RMStates_1 \times \RMStates_2$ with respect to common input alphabet $\RMEvents$ such that

\begin{enumerate}
    \item $(\RMInitialState^1, \RMInitialState^2) \in \Relation$.
    \item For every $(\RMCommonState^1, \RMCommonState^2) \in \Relation$,
    \begin{itemize}
        \item $\RMCommonState^1 \in \RMFinalStates_1$ if and only if $\RMCommonState^2 \in \RMFinalStates_2$
        \item if $\RMTransition_1(\RMCommonState^1, \RMCommonEvent) = {\RMCommonState^1}'$ for some $\RMCommonEvent \in \RMEvents$, then there exists ${\RMCommonState^2}' \in \RMStates_2$ such that $\RMTransition_2(\RMCommonState^2, \RMCommonEvent) = {\RMCommonState^2}'$ and $({\RMCommonState^1}', {\RMCommonState^2}') \in \Relation$.
        \item if $\RMTransition_2(\RMCommonState^2, \RMCommonEvent) = {\RMCommonState^2}'$ for some $\RMCommonEvent \in \RMEvents$, then there exists ${\RMCommonState^1}' \in \RMStates_1$ such that $\RMTransition_1(\RMCommonState^1, \RMCommonEvent) = {\RMCommonState^1}'$ and $({\RMCommonState^1}', {\RMCommonState^2}') \in \Relation$.
    \end{itemize}
\end{enumerate}
\end{definition}
\section{Local Labeling Functions}
\label{sec:supp_labeling_functions}

\subsection{Constructive Definition of Local Labeling Function}
\begin{definition}
\label{def:possible_team_state_set}
(Consistent set of team states) Given agent $i$'s local state $\SGCommonState_i \in \SGStates_i$ and the state $\RMCommonState^i \in \RMStates_i$ of its projected RM, we define the set $\possibleTeamStates_{\SGCommonState_i, \RMCommonState^i} \subseteq \SGJointStates \times \RMStates$ as

\[\possibleTeamStates_{\SGCommonState_i, \RMCommonState^i} = \{(\SGCommonJointState, \RMCommonState) | \SGCommonJointState \in \SGStates_1 \times \SGStates_2 \times ... \times \{\SGCommonState_i\} \times ... \times \SGStates_\SGNumAgents, \; \RMCommonState \in \RMCommonState_i \} \]
\end{definition}

For clarity, recall that any state $\RMCommonState_i$ in the projected RM $\RM_i$ corresponds to a collection of states from $\RMStates$. Thus any pair $(\SGCommonJointState, \RMCommonState) \in \possibleTeamStates_{\SGCommonState_i, \RMCommonState^i}$ corresponds to a team environment state $\SGCommonJointState \in \SGJointStates$ such that agent $i$ is in local state $\SGCommonState_i$ and to a team RM state $\RMCommonState \in \RMStates$ belonging to the collection $\RMCommonState_i \in \RMStates_i$. In words, $\possibleTeamStates_{\SGCommonState_i, \RMCommonState_i}$ is the set of all pairs of team environment states and team RM states that are consistent with the local environment state $\SGCommonState_i$ and projected RM state $\RMCommonState_i$.

\begin{definition}
\label{def:local_labeling_function}
(Local labeling function) Given a team labeling function $\RMLabelingFunction$ and a local event set $\RMEvents_i$, we define the local labeling function $\RMLabelingFunction_i$, for all $(\SGCommonState_i, \RMCommonState_i) \in \SGStates_i \times \RMStates_i$, as

\[ 
 \begin{aligned}
\RMLabelingFunction_i(\SGCommonState_i, \RMCommonState_i) =
  \begin{cases}
       \RMLabelingFunction(\SGCommonJointState, \RMCommonState) \cap \RMEvents_i, & \text{if } \exists (\SGCommonJointState, \RMCommonState) \in \possibleTeamStates_{\SGCommonState_i, \RMCommonState_i} \textrm{ such that } \RMLabelingFunction(\SGCommonJointState, \RMCommonState) \cap \RMEvents_i \neq \emptyset \\
       
        \emptyset, & \text{if } \forall (\SGCommonJointState, \RMCommonState) \in \possibleTeamStates_{\SGCommonState_i, \RMCommonState_i}, \; \RMLabelingFunction(\SGCommonJointState, \RMCommonState) \cap \RMEvents_i = \emptyset.  
  \end{cases}
  \end{aligned}
\]

\end{definition}

We define the following three conditions to ensure that $\RMLabelingFunction_i : \SGStates_i \times \RMStates_i \rightarrow \localLabelingFunctionOutput$ is well defined, and that the collection $\RMLabelingFunction_1$,..., $\RMLabelingFunction_N$ satisfy Definition \ref{def:labeling_function_decomposability}. These conditions must hold for $i = 1,2,...,N$.

\begin{enumerate}
    \item To ensure $\RMLabelingFunction_i$ maps to singleton sets of local events, or to the empty set, we must have that for any $(\SGCommonJointState, \RMCommonState) \in \SGJointStates \times \RMStates$, $|\RMLabelingFunction(\SGCommonJointState, \RMCommonState)\cap \RMEvents_i| \leq 1$.
    \item Given any input $(\SGCommonState_i, \RMCommonState_i)$, to ensure that $\RMLabelingFunction_i(\SGCommonState_i, \RMCommonState_i)$ has a unique output, there must exist a unique $\RMCommonEvent \in \RMEvents_i$ such that $\RMLabelingFunction(\SGCommonJointState, \RMCommonState) \cap \RMEvents_i = \{\RMCommonEvent\}$ or $\RMLabelingFunction(\SGCommonJointState, \RMCommonState) \cap \RMEvents_i = \emptyset$ for every  $(\SGCommonJointState, \RMCommonState) \in \possibleTeamStates_{\SGCommonState_i, \RMCommonState_i}$.
    \item For every event $\RMCommonEvent \in \RMEvents$, if $\RMCommonEvent \notin \RMLabelingFunction(\SGCommonJointState, \RMCommonState)$ then there must exists some local event set $\RMEvents_i$ containing $\RMCommonEvent$ such that $\RMLabelingFunction(\Tilde{\SGCommonJointState},\Tilde{\RMCommonState})\cap \RMEvents_i = \emptyset$ for every $(\Tilde{\SGCommonJointState},\Tilde{\RMCommonState}) \in \possibleTeamStates_{\SGCommonState_i, \RMCommonState_i}$. Here $\SGCommonState_i$ is the local state of agent $i$ consistent with team state $\SGCommonJointState$, and $\RMCommonState_i$ is the state of $\RM_i$ containing $\RMCommonState$. 
\end{enumerate}

The first condition ensures that $\RMLabelingFunction$ only returns a single event from the local event set of each agent for any step of the environment. The second condition ensures that given the local state pair $(\SGCommonState_i, \RMCommonState_i)$ of agent $i$, the same event singleton $\{\RMCommonEvent\} \in \localLabelingFunctionOutput$ is returned by $\RMLabelingFunction_i$ regardless of the states of the teammates of agent $i$. The final condition ensures that $\RMCommonEvent$ is returned by $\RMLabelingFunction$ only if it is also returned by $\RMLabelingFunction_i$ for every $i \in I_{\RMCommonEvent}$.

\subsection{Labeling Trajectories of Environment States}
\label{sec:supp_labeling_trajectories_of_states}
For any finite trajectory $\SGCommonJointState_0 \SGCommonJointState_1 ... \SGCommonJointState_k$ of team environment states, we may use the reward machine $\RM$ and labeling function $\RMLabelingFunction$ to define a sequence of triplets $(\SGCommonJointState_0, \RMCommonState_0, l_0)(\SGCommonJointState_1, \RMCommonState_1, l_1)... (\SGCommonJointState_k, \RMCommonState_k, l_k)$ and a string of events $\RMLabelingFunction(\SGCommonJointState_0 ... \SGCommonJointState_k) \in \RMEvents^*$. Here $\RMCommonState_t$ is the state of team RM $\RM$ at time $t$ and $l_t$ is the output of the labeling function $\RMLabelingFunction$ at time $t$. Algorithm \ref{alg:construct_team_triplet_sequence} details the constructive definition of the sequence.

\begin{algorithm}[h]
    \DontPrintSemicolon 
    \KwIn{$\SGCommonJointState_0 \SGCommonJointState_1 ... \SGCommonJointState_k, \RM, \RMLabelingFunction$}
    \KwOut{$(\SGCommonJointState_0, \RMCommonState_0, l_0)(\SGCommonJointState_1, \RMCommonState_1, l_1)...(\SGCommonJointState_k, \RMCommonState_k, l_k)$, $\RMLabelingFunction(\SGCommonJointState_0 \SGCommonJointState_1 ... \SGCommonJointState_k)$}
    $\RMCommonState_0 \gets \RMInitialState$, $l_0 \gets \emptyset$, $\RMEventSequence \gets emptyString()$\;
    \For{$t = 1$ \textbf{to} $k-1$} {
        $\RMCommonState_{temp} \gets \RMCommonState_t$\;
        \For{$\RMCommonEvent \in l_t$}{
            $\RMCommonState_{temp} \gets \RMTransition(\RMCommonState_{temp}, \RMCommonEvent)$\;
            $\RMEventSequence \gets append(\RMEventSequence, \RMCommonEvent)$
        }
        $\RMCommonState_{t+1} \gets \RMCommonState_{temp}$\;
        $l_{t+1} \gets \RMLabelingFunction(\SGCommonJointState_{t+1}, \RMCommonState_t)$\;
    }
    $\RMLabelingFunction(\SGCommonJointState_0 \SGCommonJointState_1 ... \SGCommonJointState_k) \gets \RMEventSequence$\;
    \Return{$(\SGCommonJointState_0, \RMCommonState_0, l_0)(\SGCommonJointState_1, \RMCommonState_1, l_1)...(\SGCommonJointState_k, \RMCommonState_k, l_k)$, $\RMLabelingFunction(\SGCommonJointState_0 \SGCommonJointState_1 ... \SGCommonJointState_k)$}\;
    \caption{Construct sequence of RM states and labeling function outputs.}
    \label{alg:construct_team_triplet_sequence}
\end{algorithm}

Similarly, for a collection of local environment states $\{\SGCommonState_0^i \SGCommonState_1^i ... \SGCommonState_k^i\}_{i=1}^N$ and the collections of projected reward machines $\{\RM_i\}_{i=1}^N$ and local labeling functions $\{\RMLabelingFunction_i\}_{i=1}^N$ we may define the collection of sequences $\{(\SGCommonState_0^i, \RMCommonState_0^i, \Tilde{l}_0^i)(\SGCommonState_1^i, \RMCommonState_1^i, \Tilde{l}_1^i) ... (\SGCommonState_k^i, \RMCommonState_k^i, \Tilde{l}_k^i)\}_{i=1}^N$ as well as a collection of strings of events $\{\RMLabelingFunction_i(\SGCommonState_0^i ... \SGCommonState_k^i)\}_{i=1}^N$. Here, $\RMCommonState_t^i$ is the state of projected RM $\RM_i$ and $\Tilde{l}_t^i$ is the output of local labeling function $\RMLabelingFunction_i$ at time $t$, after synchronization with collaborating teammates on shared events. Algorithm \ref{alg:construct_sync_triplet_sequence} details the constructive definition of the sequence.

\begin{algorithm}
    \DontPrintSemicolon 
    \KwIn{$\{\SGCommonState_0^i \SGCommonState_1^i ... \SGCommonState_k^i\}_{i=1}^N, \{\RM_i\}_{i=1}^N, \{\RMLabelingFunction_i\}_{i=1}^N$}
    \KwOut{$\{(\SGCommonState_0^i, \RMCommonState_0^i, \Tilde{l}_0^i)(\SGCommonState_1^i, \RMCommonState_1^i, \Tilde{l}_1)...(\SGCommonState_k^i, \RMCommonState_k^i, \Tilde{l}_k^i)\}_{i=1}^N$, $\{\RMLabelingFunction_i(\SGCommonState_0^i\SGCommonState_1^i ... \SGCommonState_k^i))\}_{i=1}^N$}
    \For{$i = 1$ \textbf{to} $N$}{$\RMCommonState_0^i \gets \RMInitialState^i$, $\Tilde{l}_0^i \gets \emptyset$, $\RMEventSequence_i \gets emptyString()$\;}
    \For{$t = 1$ \textbf{to} $k-1$} {
    
        \For{$i = 1$ \textbf{to} $N$}{
            $\RMCommonState_{temp}^i \gets \RMCommonState_t^i$\;
            \For{$\RMCommonEvent \in \Tilde{l}_t^i$}{
                $\RMCommonState_{temp}^i \gets \RMTransition_i(\RMCommonState_{temp}^i, \RMCommonEvent)$\;
                $\RMEventSequence_i \gets append(\RMEventSequence_i, \RMCommonEvent)$\;}
        $\RMCommonState_{t+1}^i \gets \RMCommonState_{temp}^i$\;
        $l_{t+1}^i \gets \RMLabelingFunction_i(\SGCommonState_{t+1}^i, \RMCommonState_t^i)$\;
        }
        \For{$i = 1$ \textbf{to} $N$}{
            $I_\RMCommonEvent \gets getCollaboratingAgents(l_{t+1}^i)$\;
            $\Tilde{l}_{t+1}^i \gets \bigcap_{j \in I_e} l_{t+1}^j$, (Synchronization step)\;
            
        }
    }
    \For{$i=1$ \textbf{to} $N$}{$\RMLabelingFunction_i(\SGCommonState_0^i\SGCommonState_1^i ... \SGCommonState_k^i)) \gets \RMEventSequence_i$\;}
    \Return{$\{(\SGCommonState_0^i, \RMCommonState_0^i, \Tilde{l}_0^i)(\SGCommonState_1^i, \RMCommonState_1^i, \Tilde{l}_1)...(\SGCommonState_k^i, \RMCommonState_k^i, \Tilde{l}_k^i)\}_{i=1}^N$}, $\{\RMLabelingFunction_i(\SGCommonState_0^i\SGCommonState_1^i ... \SGCommonState_k^i))\}_{i=1}^N$\;
    \caption{Construct sequence of projected RM states and synchronized labeling function outputs.}
    \label{alg:construct_sync_triplet_sequence}
\end{algorithm}

We note that the sequences $l_0 l_1 ... l_k \in (2^{\RMEvents})^*$ and $\Tilde{l}_0^i \Tilde{l}_1^i ... \Tilde{l}_k^i \in (2^{\RMEvents_i})^*$ of labeling function outputs are sequences of \textit{sets} of events. Given our assumption that $\RMLabelingFunction_i(\SGCommonState_{t+1}^i, \RMCommonState_t^i)$ outputs at most one event per time step, $|\Tilde{l}_t^i| \leq 1$ for every $t$. This assumption corresponds to the idea that only one event may occur to an individual agent per time step.

However, the team labeling function $\RMLabelingFunction$ may return multiple events at a given time step, corresponding to all the events that occur concurrently to separate agents. Because $\RM$ is assumed to be equivalent to the parallel composition of a collection of component RMs, all such concurrent events are interleaved \cite{baier2008principles}; the order in which they cause transitions in $\RM$ doesn't matter. So, Given $l_1 l_2 ... l_k \in (2^{\RMEvents})^*$, the corresponding sequence of events $\RMLabelingFunction(\SGCommonJointState_0 \SGCommonJointState_1 ... \SGCommonJointState_k)\in\RMEvents^*$ is constructed by iteratively appending elements in $l_t$ to $\RMLabelingFunction(\SGCommonJointState_0 \SGCommonJointState_1 ... \SGCommonJointState_k)$, as detailed in Algorithm \ref{alg:construct_sync_triplet_sequence}. Similarly, from sequence $\syncOutput_1^i \syncOutput_2^i ... \syncOutput_k^i \in (2^{\RMEvents_i})^*$ we construct the sequence of local events $\RMLabelingFunction_i(\SGCommonState_0^i \SGCommonState_1^i ... \SGCommonState_k^i) \in \RMEvents_i^*$. 
\newpage
\section{Extended Proof of Theorem 2}

\textbf{Theorem 2.} \textit{Given $\RM$, $\RMLabelingFunction$, and $\RMEvents_1$, ..., $\RMEvents_N$, suppose the bisimilarity condition from Theorem \ref{thm:rm_decomposability} holds. Furthermore, assume $\RMLabelingFunction$ is decomposable with respect to $\RMEvents_1$, ... $\RMEvents_N$ with the corresponding local labeling functions $\RMLabelingFunction_1$, ..., $\RMLabelingFunction_N$. Let $\SGCommonJointState_0...\SGCommonJointState_k$ be a sequence of joint environment states and $\{s_0^i...s_k^i\}_{i=1}^N$ be the corresponding sequences of local states. If the agents synchronize on shared events, then $\RM(\RMLabelingFunction(\SGCommonJointState_0...\SGCommonJointState_k)) = 1$ if and only if $\RM_i(\RMLabelingFunction_i(\SGCommonState_0^i...\SGCommonState_k^i)) = 1$ for all $i = 1,2,...,N$. Otherwise $\RM(\RMLabelingFunction(\SGCommonJointState_0...\SGCommonJointState_k)) = 0$ and $\RM_i(\RMLabelingFunction_i(\SGCommonState_0^i...\SGCommonState_k^i)) = 0$ for all $i = 1,2,...,N$.}

\begin{proof} Recall from \S \ref{sec:local_labeling_functions} it is sufficient to show that for every $i=1,2,...,N$ and every $t = 1,2,...,k$, $l_t \cap \RMEvents_i = \Tilde{l}_t^i$. Here, $l_t$ is the output of labeling function $\RMLabelingFunction$ at time $t$ and $\Tilde{l}_t^i$ is the synchronized output of local labeling function $\RMLabelingFunction_i$ (\S \ref{sec:supp_labeling_trajectories_of_states}).

At time $t=0$, $l_0$ and $\Tilde{l}_0^i$ are defined to be empty sets: no events have yet occurred (\S \ref{sec:supp_labeling_trajectories_of_states}). So, trivially $l_0 \cap \RMEvents_i = \Tilde{l}_0^i$. Furthermore, $\RMInitialState \in \RMInitialState^i$ by definition of the projected initial states $\RMInitialState^i$. Recall that $\RMInitialState \in \RMStates$ is the initial state of RM $\RM$, and $\RMInitialState^i \in \RMStates_i$ is the initial state of projected RM $\RM_i$.

Now suppose that at some arbitrary time $t$, $l_t \cap \RMEvents_i = \Tilde{l}_t^i$ and $\RMCommonState_t \in \RMCommonState_t^i$ for every $i = 1,...,N$. We wish to show that this implies $l_{t+1}\cap \RMEvents_i = \Tilde{l}_{t+1}^i$ and $\RMCommonState_{t+1} \in \RMCommonState_{t+1}^i$.

\underline{Showing $l_{t+1}\cap \RMEvents_i = \Tilde{l}_{t+1}^i$:} Recall our assumption that for any pair $(\SGCommonState_{t+1}^i, \RMCommonState_t^i)$, $\RMLabelingFunction_i(\SGCommonState_{t+1}^i, \RMCommonState_t^i)$ outputs only one event, corresponding to the idea that only one event may occur to an individual agent per time step. Thus for any $i=1,...,N$, $l_{t+1}\cap\RMEvents_i = \RMLabelingFunction(\SGCommonJointState_{t+1}, \RMCommonState_{t})\cap \RMEvents_i$ is either equal to $\{\RMCommonEvent\}$ for some $\RMCommonEvent \in \RMEvents_i$ or it is equal to the empty set.

\begin{itemize}
    \item If $\RMLabelingFunction(\SGCommonJointState_{t+1}, \RMCommonState_t) \cap \RMEvents_i = \{\RMCommonEvent\}$, then $\RMLabelingFunction_j(\SGCommonState_{t+1}^j, \RMCommonState_t^j) = \{\RMCommonEvent\}$ for every $j \in I_{\RMCommonEvent}$ by definition of $\RMLabelingFunction$ being decomposable with corresponding local labeling functions $\RMLabelingFunction_1$, $\RMLabelingFunction_2$, ..., $\RMLabelingFunction_N$. Thus $\Tilde{l}_{t+1}^i = \bigcap_{j\in I_\RMCommonEvent} \RMLabelingFunction_j(\SGCommonState_{t+1}^j, \RMCommonState_t^j) = \{\RMCommonEvent\}$. Recall that $I_e = \{i | \RMCommonEvent \in \RMEvents_i\}$.
    \item Suppose instead that $\RMLabelingFunction(\SGCommonJointState_{t+1}, \RMCommonState_t) \cap \RMEvents_i = \emptyset$.
    \begin{itemize}
        \item If $\RMLabelingFunction_i(\SGCommonState_{t+1}^i, \RMCommonState_t^i) = \emptyset$, then clearly $\Tilde{l}_{t+1}^i = \emptyset$.
        \item If $\RMLabelingFunction_i(\SGCommonState_{t+1}^i, \RMCommonState_t^i) = \{\RMCommonEvent\}$ for some $\RMCommonEvent \in \RMEvents_i$, then there exists some $j\in I_e$ such that $\RMCommonEvent \notin \RMLabelingFunction_j(\SGCommonState_{t+1}^j, \RMCommonState_{t}^j)$ by the definition of $\RMLabelingFunction$ being decomposable. Thus $\Tilde{l}_{t+1}^i = \emptyset$.
    \end{itemize}
\end{itemize}

\underline{Showing $\RMCommonState_{t+1} \in \RMCommonState_{t+1}^i$:} We begin by using the knowledge that $l_{t+1} \cap \RMEvents_i = \Tilde{l}_{t+1}^i$, and we again proceed by considering the two possible cases. 

\begin{itemize}
    \item If $l_{t+1} \cap \RMEvents_i = \Tilde{l}_{t+1}^i = \emptyset$, then the projected RM $\RM_i$ will not undergo a transition and so $\RMCommonState_{t+1}^i = \RMCommonState_t^i$. We know that $\RMCommonState_t \in \RMCommonState_{t}^i$, and that RM $\RM$ doesn't undergo any transition triggered by an event in $\RMEvents_i$. So, by definition the projected states of $\RM_i$, we have $\RMCommonState_{t+1} \in \RMCommonState_{t}^i$. Thus, $\RMCommonState_{t+1} \in \RMCommonState_{t+1}^i$.
    \item Now consider the case $l_{t+1} \cap \RMEvents_i = \Tilde{l}_{t+1}^i = \{\RMCommonEvent\}$. Projected RM $\RM_i$ will transition to a new state $\RMCommonState_{t+1}^i$ according to $\RMTransition_i(\RMCommonState_t^i, \RMCommonEvent)$. Assume, without loss of generality, that $l_{t+1}$ contains events outside of $\RMEvents_i$ which trigger transitions both before and after the transition triggered by $\RMCommonEvent$. That is, suppose $a, b \in l_{t+1} \setminus \RMEvents_i$ and that $\RM$ undergoes the following sequence of transitions: $\RMCommonState' = \RMTransition(\RMCommonState_t, a)$, $\Tilde{\RMCommonState} = \RMTransition(\RMCommonState', \RMCommonEvent)$, and finally $\RMCommonState_{t+1} = \RMTransition(\Tilde{\RMCommonState}, b)$. Because $\RMCommonState_t \in \RMCommonState_{t}^i$ and $a \notin \RMEvents_i$, we know $\RMCommonState' \in \RMCommonState_{t}^i$. Because $\RMCommonEvent \in \RMEvents_i$, $\Tilde{\RMCommonState} \in \Tilde{\RMCommonState}^i$ for some $\Tilde{\RMCommonState}^i \in \RMStates_i$ not necessarily equal to $\RMCommonState_{t}^i$. Finally, because $b \notin \RMEvents_i$, we have $\RMCommonState_{t+1} \in \Tilde{\RMCommonState}^i$. So, $\RMCommonState_{t+1} \in \Tilde{\RMCommonState}^i$ for some projected state $\Tilde{\RMCommonState}^i$ such that there exist states $\RMCommonState' \in \RMCommonState_{t}^i$ and $\Tilde{\RMCommonState} \in \Tilde{\RMCommonState}^i$ such that $\Tilde{\RMCommonState} = \RMTransition(\RMCommonState', \RMCommonEvent)$ where $\RMCommonEvent \in \RMEvents_i$. By our definition of $\RMTransition_i$ and our enforcement of it being a deterministic transition function, $\RMCommonState_{t+1}^i$ is the unique such state in $\RMStates_i$, which implies $\Tilde{\RMCommonState}^i = \RMCommonState_{t+1}^i$. Thus $\RMCommonState_{t+1} \in \RMCommonState_{t+1}^i$.
\end{itemize}

By induction we conclude the proof.

\end{proof}
\section{Rendezvous Experimental Domain}

\label{sec:supp_experiments}

Figure \ref{fig:ten_agent_rendezvous_gridworld} shows the gridworld environment for the ten-agent rendezvous task described in \S \ref{sec:experimental_results}. To successfully complete the task, all agents must simultaneously occupy the rendezvous location before proceeding to their respective goal locations. 

\begin{figure}[h!]
    \centering
    \includegraphics[width=0.5\columnwidth]{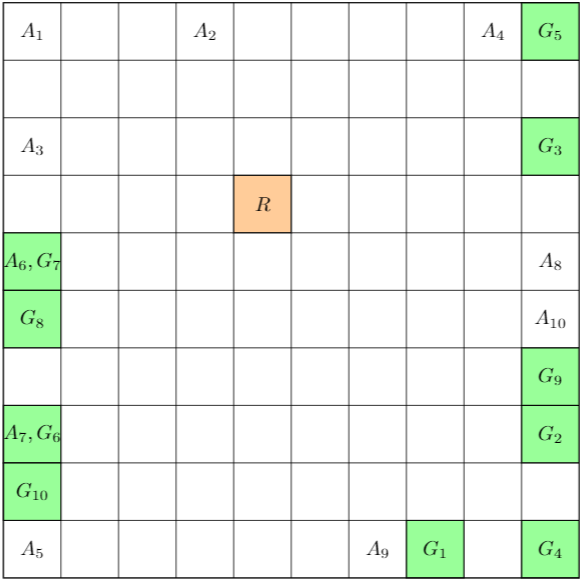}
    \caption{The initial position of agent $i$ is marked $A_i$. The common rendezvous location for all agents is marked $R$ and is highlighted in orange. The goal location for agent $i$ is marked $G_i$ and highlighted in green.}
    \label{fig:ten_agent_rendezvous_gridworld}
\end{figure}

\end{document}